\documentclass[journal]{IEEEtran}
\usepackage{amsmath,amssymb,amsthm,mathtools}
\usepackage{cite}
\usepackage{graphicx}

\usepackage[hidelinks]{hyperref}

\newtheorem{lemma}{Lemma}
\newtheorem{theorem}{Theorem}
\newtheorem{corollary}{Corollary}
\theoremstyle{remark}
\newtheorem{remark}{Remark}

\newcommand{\RegSDAC}[1]{\operatorname{Reg}_{#1}^{\mathrm{SDAC}}}

\title{Online Constrained Control of Storage Systems via Simplex Disturbance-Action
Policies}

\author{Kamiar~Asgari and Michael~J.~Neely%
\thanks{K.~Asgari and M.~J.~Neely are with the Department of Electrical and
Computer Engineering, University of Southern California, Los Angeles, CA 90089 USA
(e-mail: kamiaras@usc.edu; mjneely@usc.edu).}%
\thanks{Corresponding author: K.~Asgari.}%
}

\begin{document}
\maketitle

\begin{abstract}
We study online control of a scalar storage system with nonnegative adversarial
resource arrivals and state-dependent action constraints. The adversarial,
time-varying cost depends on both state and action. At each time, the
controller selects a feasible action before the current resource arrival and cost
are revealed. To handle the coupling between feasibility and online learning,
we introduce \emph{Simplex Disturbance-Action Control} (SDAC), whose policies
have $H$ nonnegative parameters summing to at most one. Every fixed SDAC policy
is feasible by construction. Our online SDAC controller updates these
parameters using entropic online mirror descent while preserving feasibility along
the time-varying trajectory. We prove a regret bound of
$O\!\left(\sqrt{T\log(H+1)}\right)$ relative to the best fixed SDAC policy and a
minimax lower bound showing that the dependence on $T$ is tight. We also introduce
infinite-memory SDAC policies, which include every feasible fixed-fraction policy.
SDAC policies approximate this class with an error that decreases geometrically
with $H$. For a suitable choice $H=\Theta(\log T)$,
the resulting regret against infinite-memory SDAC policies is
$O\!\left(\sqrt{T\log\log T}\right)$. The framework applies to energy-harvesting
batteries and other storage-constrained resource systems.
\end{abstract}

\begin{IEEEkeywords}
Online control, disturbance-action policies, storage systems, energy harvesting,
regret minimization, state-dependent action constraints.
\end{IEEEkeywords}

\section{Introduction}
\label{sec:introduction}

\subsection{Motivation and Problem Setting}
\IEEEPARstart{M}{any} resource-allocation systems have a storage state that changes
over time and limits the control action at each time. An action that lowers the
current cost may leave less resource for future decisions, while conserving the
resource may increase the current cost. The controller must balance these effects
while satisfying the action constraint at every time.

We study this problem for a scalar storage system that loses a fixed fraction of
its stored resource each period. At each time, the controller observes the current
state and all past arrivals and cost functions, chooses a control action, and then
observes the current arrival and cost function. The arrival and cost sequences may
be chosen adversarially, so we do not assume a known
probability model. Our goal is to construct an online SDAC controller whose
cumulative cost is close to that of the best fixed policy, chosen in hindsight,
from a specified benchmark class.

One difficulty is that the current action affects the next state. In online control,
\emph{Disturbance-Action Control} (DAC) parameterizes actions using a linear filter
of past disturbances, enabling efficient algorithms with bounded regret.
In our storage model, however, some DAC policies can produce negative actions or
actions larger than the available resource.

We address this issue by introducing \emph{Simplex Disturbance-Action Control}
(SDAC). SDAC restricts the filter parameters to a subprobability simplex and includes
powers of the retention coefficient in the filter. These choices ensure that
every fixed SDAC policy satisfies the action constraint. The state and action
along each fixed-policy trajectory are also affine in the policy parameter.
Thus, a convex state--action cost is convex in that parameter, allowing
updates through online convex optimization.

\subsection{Model and Information Pattern}

The storage state $x_t\ge0$ evolves according to
\[
x_{t+1}=\alpha x_t-u_t+w_t,
\]
where $\alpha\in(0,1)$ is a fixed retention coefficient, $u_t$ is the control
action, and $w_t\in[0,1]$ is a nonnegative resource arrival. The action must satisfy
\[
0\le u_t\le\alpha x_t.
\]
Thus, the controller cannot use more resource than remains after storage loss.

At time $t$, the controller chooses $u_t$ using the current state and the
information available through time $t-1$. The arrival $w_t$ and the convex cost
function $c_t$ are revealed after the action is selected, and the controller incurs
the cost $c_t(x_t,u_t)$. Both the online trajectory and every fixed-policy
trajectory used as a benchmark must satisfy the same dynamics and action
constraint. The formal model and assumptions are given in
Section~\ref{sec:model-setup}.

Our main application is energy-harvesting battery management, where $x_t$ is the
battery state of charge, $w_t$ is the harvested energy, $u_t$ is the energy
used, and $1-\alpha$ is the fraction of stored energy lost per period.
Similar dynamics arise in perishable inventory systems, water reservoirs with
storage loss, and queues with abandonment. In many standard queueing models,
$\alpha=1$; that case is outside our main results and is discussed in
Section~\ref{sec:extensions}.

\subsection{Contributions}

Our contributions are:
\begin{itemize}
\item \textbf{Policy class.}
    We introduce SDAC policies parameterized by a
    subprobability simplex. We prove that every such fixed policy is feasible
    for every arrival sequence with entries in $[0,1]$.

    We also introduce infinite-memory SDAC policies in
    Section~\ref{subsec:infinite-memory-sdac}. Their parameter domain extends the
    subprobability simplex to nonnegative sequences whose sum is at most one.
    This class contains every SDAC policy and every feasible
    fixed-fraction policy $u_t=kx_t$ with $k\in[0,\alpha]$.
    Truncating the parameter sequence gives an SDAC policy whose
    approximation error decreases geometrically with the memory horizon $H$.

\item \textbf{Online control and regret.}
    We develop an online SDAC controller that updates its parameter using
    entropic online mirror descent (OMD) over the subprobability simplex.
    We prove that the controller remains feasible and achieves
    $O\!\left(\sqrt{T\log(H+1)}\right)$
    regret against the best fixed SDAC policy.
    For a suitable choice $H=\Theta(\log T)$, the regret against
    infinite-memory SDAC policies is
    $O\!\left(\sqrt{T\log\log T}\right)$. The same bound
    holds against the best fixed-fraction policy.

\item \textbf{Minimax lower bound.}
    We prove a lower bound of order $\sqrt{T}$ for the expected regret of every
    causal feasible controller, including randomized controllers; see
    Theorem~\ref{thm:sdac-lower-bound}. When $T$ is at least on the order of
    $(1-\alpha)^{-1}$, the lower bound scales as
    $\Omega((1-\alpha)^{-1}\sqrt{T})$ near $\alpha=1$.
    The upper and lower bounds differ by a factor of $(1-\alpha)^{-1/2}$
    in their dependence on $1-\alpha$.
    Closing this gap is left for future work.
\end{itemize}

\subsection{Related Work}

\paragraph{Online control}
Agarwal et al.\ use
DAC policies to obtain $\widetilde O(\sqrt{T})$ regret against strongly
stable linear controllers \cite{agarwal2019online}. More closely related to
our simplex parameterization, Golowich et al.\ adapt DAC to population
dynamics and search over simplex-structured DAC policy parameters
\cite{golowich2024population}. Their state lies in a probability simplex,
and their feasible action set is independent of the current state. Here,
the constraint $0\le u_t\le\alpha x_t$ couples each action to the available
resource.

Li et al.\ consider affine constraints on the state and action
\cite{li2021onlineConstrained}. Their action constraint is independent of
the current state. They enforce feasibility by restricting and projecting
their policy updates onto a tightened parameter set. In our work, the
action produced by every fixed SDAC policy is feasible by construction.
When the parameter changes online, we clip the nominal action to the
current admissible interval as in \eqref{eq:alg-feasible-action}.

Liu et al.\ consider general convex constraint functions
\cite{liu2023constrainedControl}. For their adversarial constraints, the
guarantees bound cumulative violation but may permit violations at individual
times. Here,
\eqref{eq:model-feasibility} must hold at every time.

Kumar et al.\ introduce online convex optimization with unbounded memory
\cite{kumar2023unbounded}. For their discounted infinite-memory model, the
lower bound scales as $\sqrt{T}$ divided by one minus the discount factor. This
parallels our $\Omega((1-\alpha)^{-1}\sqrt{T})$ lower bound beyond the mixing
time, although ours applies directly to causal feasible
controllers under the state-dependent action constraint.

\paragraph{Energy harvesting}
Fixed-fraction and related allocation policies for stochastic energy arrivals
motivate our infinite-memory SDAC policies. Shaviv and \"Ozg\"ur
\cite{shaviv2016universally} allocate a fixed fraction of the available energy
under i.i.d.\ arrivals. Arafa et al.\ \cite{arafa2018online} extend fixed-fraction
policies to general concave utilities under i.i.d.\ arrivals and, for Bernoulli
arrivals, characterize a non-increasing allocation sequence following a full
recharge. Zibaeenejad and Chen \cite{zibaeenejad2019lookahead} derive an optimal
decreasing allocation sequence under Bernoulli arrivals when the controller
observes a fixed look-ahead window. Our infinite-memory SDAC policies generalize
this allocation structure. The class allows decreasing and non-increasing
allocation sequences as special cases; see Section~\ref{subsec:infinite-memory-sdac} and
Corollary~\ref{cor:fraction-as-sdac}.

Yu and Neely \cite{Yu2019} study expected utility maximization under i.i.d.\ system
states, while Asgari and Neely \cite{10.1145/3428337} obtain an
$O(\sqrt{T})$ regret--battery-capacity tradeoff for arbitrary loss sequences under
i.i.d.\ energy arrivals. Neither objective depends directly on the battery state.
Yu and Neely optimize expected time-average utility, whereas Asgari and Neely
measure regret against a fixed action. Lin et al.\ \cite{lin2025optimal} allow
arbitrary energy arrivals and time-varying age-of-information
objectives. They observe the current arrival before allocation and compare with an
unrestricted offline optimum through a competitive ratio. In our setting, the
controller acts before observing the current arrival and cost, the arrival and cost
sequences may be chosen adversarially, the cost depends on the storage state and
action, and the benchmark is a fixed SDAC policy chosen in hindsight.

\paragraph{Queueing and network optimization}
Classical queueing control studies throughput and queue stability through dynamic
service and scheduling decisions
\cite{tassiulas1990stability,tassiulas1993dynamic,mckeown1999achieving}, while
stochastic network optimization studies tradeoffs between time-average performance
and average queue backlog \cite{neely2010stochastic}. Our performance criterion is
finite-horizon regret against the best fixed SDAC policy in hindsight.

For their time-invariant class, Khojastepour and Sabharwal
\cite{khojastepour2004delay} represent schedulers satisfying a maximum-delay
constraint as fixed linear filters of the current and past arrivals, with
nonnegative coefficients summing to one rather than at most one as in SDAC.
Under i.i.d.\ arrivals, they identify the uniform moving-average filter as the
minimizer of long-run average transmit power over a Gaussian channel. The
connection is limited to the simplex-like filter structure: they solve a
stochastic average-power problem with a fixed cost, whereas our algorithm
updates the SDAC parameters online under adversarial arrival and cost sequences
with general convex state--action costs.

\paragraph{Inventory and reservoir management}
Yang et al.\ \cite{yang2020online} study online procurement under inventory
constraints and obtain optimal competitive ratios relative to an offline optimum.
They observe the current demand and price before acting, whereas our controller
acts before the current arrival and cost are revealed. Hihat and Fermanian
\cite{hihat2024online} learn a base-stock policy for more general inventory
dynamics but do not prove a regret bound for their method. For our scalar linear
model, we prove feasibility, approximation, and sublinear regret guarantees.
Classical stochastic reservoir models study storage under random inflows and fixed
release rules \cite{moran1954dams}, with an emphasis on steady-state behavior rather
than finite-horizon online regret.

\section{Model and Regret Benchmark}
\label{sec:model-setup}

For every positive integer $n$, we write
$[n]\triangleq\{1,\ldots,n\}$. For every $z\in\mathbb{R}$, we define
$[z]_+\triangleq\max\{z,0\}$. For
$\mathbf{v}\in\mathbb{R}^n$, the $i$-th coordinate is denoted by $v_i$. If the
vector already carries a subscript, the coordinate is appended after a comma: the
$i$-th coordinate of $\mathbf{v}_t$ is $v_{t,i}$. We write
$\mathbf{1}$ for the all-ones vector, with dimension clear from context, and
$\langle\cdot,\cdot\rangle$ for the Euclidean inner product. We use $\log$ for the
natural logarithm.

Fix a positive integer $T$ and a known retention coefficient
$\alpha\in(0,1)$. The state $x_t\in\mathbb{R}_{\ge0}$ is the resource level at the
beginning of time $t$. The system starts empty:
\begin{equation}
\label{eq:model-init}
x_1=0.
\end{equation}
At each time $t\in[T]$, the controller chooses an action
$u_t\in\mathbb{R}_{\ge0}$ satisfying
\begin{equation}
\label{eq:model-feasibility}
0\le u_t\le\alpha x_t.
\end{equation}
The system then receives an arrival $w_t\in[0,1]$ and evolves according to
\begin{equation}
\label{eq:model-dynamics}
x_{t+1}=\alpha x_t-u_t+w_t.
\end{equation}
Here, $1-\alpha$ is the fraction of stored resource lost per period.
The constraint \eqref{eq:model-feasibility} ensures that the action does not exceed
the amount remaining after this loss.

In the online-control literature, $w_t$ is commonly called a disturbance. Here, we
call it an \emph{arrival} because it is nonnegative, while retaining the standard
name Disturbance-Action Control. Nonnegativity is important for both the storage
interpretation and the
feasibility properties of SDAC.

At time $t$, the controller chooses $u_t$ before observing $w_t$ or $c_t$. After the
action is selected, the arrival and the entire cost function are revealed, the
controller incurs $c_t(x_t,u_t)$, and the state is updated according to
\eqref{eq:model-dynamics}.

For every feasible trajectory, \eqref{eq:model-dynamics} and $w_t\in[0,1]$ give
\[
0\le x_{t+1}\le\alpha x_t+1.
\]
Since $x_1=0$, induction gives
\[
0\le x_t
\le\frac{1-\alpha^{t-1}}{1-\alpha}
\le\frac{1}{1-\alpha},
\qquad t\in[T+1].
\]
Together with \eqref{eq:model-feasibility}, this bound places every feasible
state--action pair in the compact set
\begin{equation}
\label{eq:model-domain}
\mathcal{D}_\alpha
\triangleq
\left\{
(x,u):
0\le x\le\frac{1}{1-\alpha},
0\le u\le\alpha x
\right\}.
\end{equation}

\begin{remark}[Storage capacity]
A storage capacity of $1/(1-\alpha)$ suffices to avoid overflow.
\end{remark}

Each revealed cost function $c_t:\mathbb{R}_{\ge0}^2\to\mathbb{R}$ is convex in
$(x,u)$ and satisfies the following Lipschitz condition on $\mathcal{D}_\alpha$:
there are constants $L_x,L_u>0$ such that
\begin{equation}
\label{eq:model-lipschitz}
|c_t(x,u)-c_t(x',u')|
\le
L_x|x-x'|+L_u|u-u'|
\end{equation}
for all $(x,u),(x',u')\in\mathcal{D}_\alpha$.
At every $(x,u)\in\mathcal{D}_\alpha$, we also assume access to a cost subgradient
\begin{equation}
\label{eq:model-cost-subgradient}
\boldsymbol{\zeta}_t(x,u)
=
\bigl(\zeta_t^x(x,u),\zeta_t^u(x,u)\bigr)
\in\partial c_t(x,u)
\end{equation}
satisfying
\[
|\zeta_t^x(x,u)|\le L_x,
\qquad
|\zeta_t^u(x,u)|\le L_u.
\]
These subgradient bounds hold, for example, when $c_t$ has a convex
extension to an open neighborhood of $\mathcal{D}_\alpha$ satisfying
\eqref{eq:model-lipschitz}; see
\cite[Thm.~9.13]{rockafellar1998variational}.

We use $x_t$ and $u_t$ without parameter arguments for the trajectory generated by
the online controller. Its cumulative cost is
\begin{equation}
\label{eq:def-cumulative-cost}
J_T^{\mathrm{on}}
\triangleq
\sum_{t=1}^T c_t(x_t,u_t).
\end{equation}

Throughout, SDAC policies have finite memory unless explicitly described as
infinite-memory SDAC policies.
The parameter domain $\mathcal{M}_H$ and the infinite-memory
parameter domain $\mathcal{M}_\infty$ are defined in Section~\ref{sec:sdac}.
For a fixed parameter $\mathbf{m}$, the corresponding \emph{fixed-policy
trajectory} is the unique state--action
sequence $\{(x_t(\mathbf{m}),u_t(\mathbf{m}))\}_{t=1}^T$ obtained by applying
the SDAC policy with parameter $\mathbf{m}$ at every time.
Its cumulative cost is
\[
J_T(\mathbf{m})
\triangleq
\sum_{t=1}^T c_t\bigl(x_t(\mathbf{m}),u_t(\mathbf{m})\bigr).
\]
For a fixed positive integer $H$, the regret against the best fixed
SDAC policy is
\begin{equation}
\label{eq:def-regret}
\RegSDAC{T}(H)
\triangleq
J_T^{\mathrm{on}}
-\min_{\mathbf{m}\in\mathcal{M}_H}J_T(\mathbf{m}).
\end{equation}
Our objective is to satisfy \eqref{eq:model-feasibility} at every time and achieve
sublinear regret. We use the same trajectory and cost notation for
infinite-memory SDAC policies indexed by $\mathbf{m}\in\mathcal{M}_\infty$.

For fixed-fraction policies, $x_t(k)$ and $u_t(k)$ denote the state and action
induced by $u_t(k)=kx_t(k)$ for $k\in[0,\alpha]$.

\section{Simplex Disturbance-Action Control (SDAC)}
\label{sec:sdac}

\subsection{From Disturbance-Action Control to SDAC}
DAC uses a finite linear filter of past disturbances. Under standard conditions, this
parameterization yields convex surrogate costs and efficient online updates. In
our scalar storage model, DAC has the following form:
\begin{equation}
\label{eq:dac-param}
u_t = \sum_{i=1}^{H} \theta_i w_{t-i},
\end{equation}
where $H$ is a positive integer (the memory horizon) and
$\boldsymbol{\theta}=(\theta_1,\ldots,\theta_H)^\top\in\mathbb{R}^H$ is the vector
of filter coefficients.

For our constrained model, an unrestricted filter does not ensure feasibility: its
action can exceed $\alpha x_t$ or become negative.
We therefore introduce \emph{Simplex Disturbance-Action Control} (SDAC), which uses a
simplex-constrained filter tailored to the storage dynamics. This restriction makes
every fixed SDAC policy feasible while keeping surrogate costs convex and online
updates efficient.

\subsection{Policies and Structural Properties}
Fix a memory horizon $H\in\mathbb{N}_{>0}$ and define the subprobability simplex
\begin{equation}
\label{eq:def-simplex}
\mathcal{M}_H
\triangleq
\left\{
\mathbf{m}\in\mathbb{R}_{\ge0}^H:
\|\mathbf{m}\|_1\le1
\right\}.
\end{equation}
With the convention $w_\tau\equiv0$ for $\tau\le0$, each
$\mathbf{m}\in\mathcal{M}_H$ defines an SDAC policy through
\begin{equation}
\label{eq:def-sdac-policy}
u_t(\mathbf{m})
\;\triangleq\;
\sum_{i=1}^{H} m_i\alpha^i w_{t-i}
=
\langle\boldsymbol{\phi}_t^u,\mathbf{m}\rangle,
\end{equation}
where
\begin{equation}
\label{eq:def-phi-u}
\boldsymbol{\phi}_t^u
\triangleq
\bigl(\alpha w_{t-1},\alpha^2 w_{t-2},\ldots,\alpha^H w_{t-H}\bigr)^\top
\in\mathbb{R}^H.
\end{equation}
Here, $m_i\alpha^i w_s$ is the amount from arrival $w_s$ used $i$ periods later.

The following state decomposition yields two structural properties: every
fixed SDAC policy satisfies the feasibility constraint
\eqref{eq:model-feasibility}, and its induced cost is convex and Lipschitz in
$\mathbf{m}$. Omitted proofs appear in
Appendix~\ref{sec:appendix-sdac-proofs}.

\begin{lemma}[State decomposition under SDAC]
\label{lem:sdac-state-decomp}
Fix a positive integer $H$.
For $t\in[T]$ and $i\in[H]$, define
\[
r_t(i)
\triangleq
\sum_{j=1}^{i}\alpha^{j-1}w_{t-j},
\qquad
r_t
\triangleq
\sum_{j=1}^{t-1}\alpha^{j-1}w_{t-j},
\]
where the empty sum defining $r_1$ is zero. Then, for every
$\mathbf{m}\in\mathcal{M}_H$ and $t\in[T]$, the fixed-policy state associated
with $\mathbf{m}$ satisfies
\begin{equation}
\label{eq:sdac-state-decomp}
x_t(\mathbf{m})
=
\sum_{i=1}^H m_i r_t(i)
+\bigl(1-\sum_{i=1}^H m_i\bigr)r_t.
\end{equation}
Equivalently, with
\[
\boldsymbol{\phi}_t^x
\triangleq
\bigl(r_t(1)-r_t,\ldots,r_t(H)-r_t\bigr)^\top,
\]
we have
\[
x_t(\mathbf{m})
=
r_t+\langle\boldsymbol{\phi}_t^x,\mathbf{m}\rangle.
\]
Thus, $\mathbf{m}\mapsto x_t(\mathbf{m})$ is affine.
\end{lemma}

\begin{lemma}[Feasibility of SDAC policies]
\label{lem:sdac-feasibility}
For every $\mathbf{m}\in\mathcal{M}_H$, the fixed SDAC policy
indexed by $\mathbf{m}$ is feasible:
\[
0\le u_t(\mathbf{m})\le \alpha x_t(\mathbf{m})
\qquad \text{for every }t\in[T].
\]
\end{lemma}

\begin{proof}
Fix $t\in[T]$. For every $i\in[H]$,
\[
\alpha r_t(i)-\alpha^i w_{t-i}
=\sum_{j=1}^{i-1}\alpha^j w_{t-j}
\ge0.
\]
Since $\mathbf{m}\in\mathcal{M}_H$, we have $m_i\ge0$ for every $i\in[H]$ and
$1-\sum_{i=1}^H m_i\ge0$. Also, $w_{t-i}\ge0$ for every $i\in[H]$ and
\[
r_t=\sum_{j=1}^{t-1}\alpha^{j-1}w_{t-j}\ge0.
\]
Therefore, using \eqref{eq:sdac-state-decomp} in the last step,
\begin{align*}
0\le u_t(\mathbf{m})
&=\sum_{i=1}^H m_i\alpha^i w_{t-i}
\le\sum_{i=1}^H m_i\alpha r_t(i) \\
&\le\sum_{i=1}^H m_i\alpha r_t(i)
+\alpha\left(1-\sum_{i=1}^H m_i\right)r_t \\
&=\alpha x_t(\mathbf{m}).\qedhere
\end{align*}
\end{proof}

\begin{lemma}[Convexity and Lipschitz continuity of the surrogate cost]
\label{lem:sdac-surrogate}
For $t\in[T]$, define
\[
\ell_t(\mathbf{m})
\triangleq
c_t\bigl(x_t(\mathbf{m}),u_t(\mathbf{m})\bigr),
\qquad
\mathbf{m}\in\mathcal{M}_H,
\]
and
\begin{equation}
\label{eq:def-surrogate-G}
G_\alpha
\triangleq
\alpha\left(\frac{L_x}{1-\alpha}+L_u\right).
\end{equation}
Then $\ell_t$ is convex on $\mathcal{M}_H$ and, for all
$\mathbf{m},\mathbf{m}'\in\mathcal{M}_H$,
\[
|\ell_t(\mathbf{m})-\ell_t(\mathbf{m}')|
\le
G_\alpha\|\mathbf{m}-\mathbf{m}'\|_1.
\]
For $\mathbf{m}\in\mathcal{M}_H$, use \eqref{eq:model-cost-subgradient} to define
\begin{equation}
\label{eq:surr-subgradient}
\begin{aligned}
(\zeta_t^x,\zeta_t^u)
&=
\boldsymbol{\zeta}_t\bigl(x_t(\mathbf{m}),u_t(\mathbf{m})\bigr), \\
\mathbf{g}_t
&\triangleq
\zeta_t^x\boldsymbol{\phi}_t^x
+\zeta_t^u\boldsymbol{\phi}_t^u.
\end{aligned}
\end{equation}
Then
\[
\mathbf{g}_t\in\partial\ell_t(\mathbf{m}),
\qquad
\|\mathbf{g}_t\|_\infty\le G_\alpha.
\]
\end{lemma}

\subsection{Approximation of Infinite-Memory SDAC Policies}
\label{subsec:infinite-memory-sdac}

Define the infinite-memory parameter domain
\begin{equation}
\label{eq:def-infinite-simplex}
\mathcal{M}_\infty
\triangleq
\left\{
\mathbf{m}=(m_i)_{i\ge1}:m_i\ge0,\quad
\sum_{i=1}^{\infty}m_i\le1
\right\}.
\end{equation}
Each $\mathbf{m}\in\mathcal{M}_\infty$ defines an infinite-memory SDAC policy
through
\begin{equation}
\label{eq:def-infinite-sdac-policy}
u_t(\mathbf{m})
\triangleq
\sum_{i=1}^{t-1}m_i\alpha^i w_{t-i}.
\end{equation}
The state $x_t(\mathbf{m})$ follows \eqref{eq:model-dynamics} with initial
condition \eqref{eq:model-init}.
After extending vectors in $\mathcal{M}_H$ by zeros, SDAC policies with horizon $H$
are exactly the infinite-memory SDAC policies with $m_i=0$ for every $i>H$.

\begin{lemma}[Approximation of infinite-memory SDAC policies]
\label{lem:sdac-truncation-approx}
Every infinite-memory SDAC policy is feasible for every arrival sequence with
entries in $[0,1]$. Fix $\mathbf{m}\in\mathcal{M}_\infty$ and an integer $H\ge1$,
and define its truncation by
\[
\mathbf{m}_H(\mathbf{m})\triangleq(m_1,\ldots,m_H)^\top.
\]
Then $\mathbf{m}_H(\mathbf{m})\in\mathcal{M}_H$. If each $c_t$ satisfies
\eqref{eq:model-lipschitz}, then, for every such arrival sequence,
\begin{equation}
\label{eq:sdac-truncation-cost-gap}
\begin{aligned}
&\left|J_T\bigl(\mathbf{m}_H(\mathbf{m})\bigr)-J_T(\mathbf{m})\right| \\
&\qquad\le
\left(L_u+\frac{L_x}{1-\alpha}\right)\alpha^{H+1}T.
\end{aligned}
\end{equation}
\end{lemma}

\begin{corollary}[Fixed-fraction policies as infinite-memory SDAC policies]
\label{cor:fraction-as-sdac}
For $k\in[0,\alpha]$, define
\[
m_i(k)=\frac{k}{\alpha}\left(1-\frac{k}{\alpha}\right)^{i-1},
\qquad i\ge1,
\]
where $0^0=1$.
Then $\mathbf{m}(k)\in\mathcal{M}_\infty$ and, for every arrival sequence with
entries in $[0,1]$, the infinite-memory SDAC policy indexed by $\mathbf{m}(k)$
coincides with the feasible fixed-fraction policy $u_t(k)=kx_t(k)$:
\[
\bigl(x_t(\mathbf{m}(k)),u_t(\mathbf{m}(k))\bigr)
=\bigl(x_t(k),u_t(k)\bigr),
\qquad t\in[T].
\]
\end{corollary}

\section{Online SDAC Controller}
\label{sec:online-algorithm}

The controller maintains a parameter $\mathbf{m}_t\in\mathcal{M}_H$ and updates
it by OMD \cite{hazan2016introduction} with step size $\eta>0$ and unnormalized
negentropy:
\begin{equation}
\label{eq:def-negentropy}
\Psi(\mathbf{m}) \triangleq \sum_{i=1}^H\bigl(m_i\log m_i-m_i\bigr), \qquad
\mathbf{m}\in\mathbb{R}_{\ge0}^H,
\end{equation}
where $0\log0\triangleq0$. For $\mathbf{m}\in\mathbb{R}_{\ge0}^H$ and
$\mathbf{v}\in\mathbb{R}_{>0}^H$, the associated Bregman divergence is
\begin{equation}
\label{eq:def-bregman}
D_\Psi(\mathbf{m}\|\mathbf{v})
\triangleq
\sum_{i=1}^H\Bigl(m_i\log\tfrac{m_i}{v_i}-m_i+v_i\Bigr).
\end{equation}
Initialize the online parameter as
\begin{equation}
\label{eq:alg-init}
\mathbf{m}_1\triangleq\frac{1}{H+1}\mathbf{1}.
\end{equation}

\paragraph{Action selection}
At time $t$, the controller uses $\mathbf{m}_t$, initialized by
\eqref{eq:alg-init} for $t=1$ and produced by the OMD update at time $t-1$
for $t\ge2$. Fixed-policy feasibility need not hold when $\mathbf{m}_t$ varies: the
nominal action $\langle\boldsymbol{\phi}_t^u,\mathbf{m}_t\rangle$ can exceed
$\alpha x_t$. The controller therefore selects
\begin{equation}
\label{eq:alg-feasible-action}
u_t
\triangleq
\min\!\left\{
\langle\boldsymbol{\phi}_t^u,\mathbf{m}_t\rangle,
\alpha x_t
\right\}.
\end{equation}
The state then evolves according to \eqref{eq:model-dynamics}.
Since $\boldsymbol{\phi}_t^u$ and $\mathbf{m}_t$ are componentwise nonnegative,
$\langle\boldsymbol{\phi}_t^u,\mathbf{m}_t\rangle\ge0$. If $x_t\ge0$, then
\eqref{eq:alg-feasible-action} gives $0\le u_t\le\alpha x_t$, and
\eqref{eq:model-dynamics} gives $x_{t+1}\ge w_t\ge0$. Since $x_1=0$, induction
shows that the online trajectory is feasible.

Lemma~\ref{lem:availability-projection} shows that the $D_\Psi$-projection of
$\mathbf{m}_t$ onto the parameters compatible with $x_t$ induces exactly the
clipped action $u_t$.
\paragraph{Surrogate cost and OMD update}
After $w_t$ and $c_t$ are revealed, the controller evaluates the subgradient
$\mathbf{g}_t$ in \eqref{eq:surr-subgradient} at $\mathbf{m}_t$ and sets
\begin{equation}
\label{eq:alg-omd-argmin}
\mathbf{m}_{t+1}
\triangleq
\operatorname*{arg\,min}_{\mathbf{m}\in\mathcal{M}_H}
\left\{
\eta\langle \mathbf{g}_t,\mathbf{m}\rangle
+D_\Psi(\mathbf{m}\|\mathbf{m}_t)
\right\}.
\end{equation}
Lemma~\ref{lem:omd-closed-form} gives a two-step implementation. First, take the
coordinatewise multiplicative step
\begin{equation}
\label{eq:alg-omd-mult}
\widetilde{m}_{t+1,i}
\triangleq
m_{t,i}\exp\!\bigl(-\eta g_{t,i}\bigr),
\qquad
i\in[H].
\end{equation}
Next, project onto $\mathcal{M}_H$ by rescaling if needed:
\begin{equation}
\label{eq:alg-omd-proj}
\begin{aligned}
\mathbf{m}_{t+1}
&=\operatorname*{arg\,min}_{\mathbf{m}\in\mathcal{M}_H}
D_\Psi(\mathbf{m}\|\widetilde{\mathbf{m}}_{t+1}) \\
&=\frac{\widetilde{\mathbf{m}}_{t+1}}
{\max\{1,\|\widetilde{\mathbf{m}}_{t+1}\|_1\}}.
\end{aligned}
\end{equation}

\section{Regret Upper Bounds}
\label{sec:regret-analysis}
The next theorem bounds the regret of the online SDAC controller against the
best fixed SDAC policy. Proofs appear in
Appendix~\ref{sec:appendix-regret-proofs}.

\begin{theorem}[Regret against fixed SDAC policies]
\label{thm:regret-bound}
Suppose the costs satisfy the assumptions of Section~\ref{sec:model-setup}.
For any step size $\eta>0$, the controller in
Section~\ref{sec:online-algorithm}, initialized by \eqref{eq:alg-init}, satisfies
\begin{equation}
\label{eq:regret-theorem-bound}
\RegSDAC{T}(H)
\le
\frac{\log(H+1)}{\eta}+\eta T\Gamma_\alpha,
\end{equation}
where
\begin{equation}
\label{eq:def-Gamma-alpha}
\Gamma_\alpha
\triangleq
\frac{G_\alpha^2}{2}
+
\frac{\alpha G_\alpha(L_x+\alpha L_u)}{(1-\alpha)^2},
\end{equation}
and $G_\alpha=\alpha\bigl(\frac{L_x}{1-\alpha}+L_u\bigr)$ as in
\eqref{eq:def-surrogate-G}.

Choosing $\eta^*=\sqrt{\log(H+1)/(T\Gamma_\alpha)}$ gives
\[
\RegSDAC{T}(H)\le 2\sqrt{\Gamma_\alpha T\log(H+1)}.
\]
\end{theorem}

The proof combines the OMD bound for the surrogate costs
(Lemma~\ref{lem:omd-surrogate-regret}) with the following bound on the
difference between realized and surrogate costs.
\begin{lemma}[Stability under online parameter updates]
\label{lem:contractive-deviation}
Under the assumptions of Theorem~\ref{thm:regret-bound}, the online SDAC
controller satisfies
\begin{equation}
\label{eq:stability-sum-bound}
\sum_{t=1}^T \Bigl(c_t(x_t,u_t)-\ell_t(\mathbf{m}_t)\Bigr)
\le
\frac{\alpha\eta T G_\alpha(L_x+\alpha L_u)}{(1-\alpha)^2}.
\end{equation}
\end{lemma}

\begin{corollary}[No-regret against infinite-memory SDAC policies]
\label{cor:regret-vs-infinite-sdac}
Under the assumptions of Theorem~\ref{thm:regret-bound}, run the online SDAC
controller with step size $\eta^*$ and memory horizon
\[
H=H_T
\triangleq
\max\!\left\{1,\left\lceil\frac{\log T}{\log(1/\alpha)}\right\rceil\right\},
\]
where $\eta^*$ is given in that theorem. Define the regret against the best fixed
infinite-memory SDAC policy by
\[
\RegSDAC{T}(\infty)
\triangleq
J_T^{\mathrm{on}}-\inf_{\mathbf{m}\in\mathcal{M}_\infty}J_T(\mathbf{m}).
\]
Then
\begin{equation}
\label{eq:infinite-sdac-regret-bound}
\begin{aligned}
\RegSDAC{T}(\infty)
&\le
2\sqrt{\Gamma_\alpha T\log(H_T+1)}\\
&\quad+
\left(L_u+\frac{L_x}{1-\alpha}\right)\alpha^{H_T+1}T.
\end{aligned}
\end{equation}
For fixed $\alpha,L_x,L_u$, the right-hand side is
$O\!\left(\sqrt{T\log\log T}\right)=o(T)$.
The same bound holds against the best fixed-fraction policy, since
Corollary~\ref{cor:fraction-as-sdac} gives
\[
J_T^{\mathrm{on}}
-\min_{k\in[0,\alpha]}\sum_{t=1}^T c_t\bigl(x_t(k),u_t(k)\bigr)
\le\RegSDAC{T}(\infty).
\]
\end{corollary}

\section{Minimax Regret Lower Bound}
\label{sec:minimax-lower-bound}

The next result shows that the $\sqrt{T}$ dependence is unavoidable and
quantifies the necessary dependence on $1-\alpha$ for any
causal feasible controller. Omitted proofs appear in
Appendix~\ref{sec:appendix-minimax-proofs}.

\begin{theorem}[Minimax SDAC regret lower bound]
\label{thm:sdac-lower-bound}
Fix $H\ge1$ and $T\ge2$. For every causal feasible online controller
$\mathcal{A}$, possibly randomized, there is a deterministic oblivious sequence
of arrivals and costs satisfying the assumptions of
Section~\ref{sec:model-setup} such that
\begin{equation}
\label{eq:sdac-lower-exact}
\mathbb{E}_{\mathcal{A}}\!\left[\RegSDAC{T}(H)\right]
\ge
\frac{1}{2\sqrt{2}}
\sqrt{\sum_{t=2}^T \beta_t^2},
\end{equation}
where
\begin{equation}
\label{eq:sdac-lower-beta}
\beta_t
\triangleq
\alpha\left(
L_u+\frac{L_x(1-\alpha^{t-2})}{1-\alpha}
\right),
\qquad t=2,\ldots,T.
\end{equation}
The expectation is over the internal randomization of $\mathcal{A}$.
\end{theorem}

\begin{corollary}[Lower bound beyond the mixing time]
\label{cor:sdac-lower-mixing}
Fix $H\ge1$ and define
\[
t_{\mathrm{mix}}(\alpha)
\triangleq
\left\lceil
\frac{\log 2}{\log(1/\alpha)}
\right\rceil.
\]
If $T\ge2t_{\mathrm{mix}}(\alpha)+2$, then for every causal feasible controller
$\mathcal{A}$, possibly randomized, there is a deterministic oblivious sequence
satisfying the model assumptions such that
\[
\mathbb{E}_{\mathcal{A}}\!\left[\RegSDAC{T}(H)\right]
\ge
\frac{G_\alpha}{8}\sqrt{T},
\]
where $G_\alpha=\alpha(L_x/(1-\alpha)+L_u)$ as in
\eqref{eq:def-surrogate-G}.
\end{corollary}

\begin{remark}[Large-$\alpha$ regime]
\label{rem:large-alpha-regime}
For fixed $L_x,L_u>0$, both $t_{\mathrm{mix}}(\alpha)$ and $G_\alpha$ are of
order $(1-\alpha)^{-1}$ as $\alpha\uparrow1$. If
$T\ge2t_{\mathrm{mix}}(\alpha)+2$, then
Corollary~\ref{cor:sdac-lower-mixing} yields
\[
\mathbb{E}_{\mathcal{A}}\!\left[\RegSDAC{T}(H)\right]
=
\Omega\!\left(\frac{\sqrt{T}}{1-\alpha}\right).
\]
For the upper bound, as $\alpha\uparrow1$,
\[
(1-\alpha)G_\alpha
=\alpha\bigl(L_x+(1-\alpha)L_u\bigr)
\longrightarrow L_x,
\]
and \eqref{eq:def-Gamma-alpha} gives
\[
\begin{aligned}
(1-\alpha)^3\Gamma_\alpha
&=\frac{1-\alpha}{2}\bigl((1-\alpha)G_\alpha\bigr)^2 \\
&\quad+\alpha\bigl((1-\alpha)G_\alpha\bigr)
(L_x+\alpha L_u) \\
&\longrightarrow L_x(L_x+L_u)>0.
\end{aligned}
\]
Thus,
\[
\sqrt{\Gamma_\alpha}
=\Theta\!\left((1-\alpha)^{-3/2}\right)
\quad\text{as }\alpha\uparrow1.
\]
The upper bound is $O((1-\alpha)^{-3/2}\sqrt{T\log(H+1)})$
by Theorem~\ref{thm:regret-bound}. Its dependence on $1-\alpha$ exceeds that of
the lower bound by a factor of $(1-\alpha)^{-1/2}$.

The lower-bound scaling requires $T$ to grow with $t_{\mathrm{mix}}(\alpha)$.
For fixed $T$,
\[
\frac{1-\alpha^{t-2}}{1-\alpha}\longrightarrow t-2
\qquad\text{as }\alpha\uparrow1,
\]
so the lower bound in \eqref{eq:sdac-lower-exact} remains finite.
\end{remark}

\section{Discussion and Extensions}
\label{sec:extensions}

We discuss three extensions. In each case, the feasibility, convexity, and
stability arguments apply with the stated changes to the policy domain and
Lipschitz constants. We omit the resulting regret bounds, whose form is unchanged
up to these constants and, in the vector settings, the chosen product norm.

\paragraph{Retention $\alpha\ge 1$}
For $\alpha\ge1$, fix $k$ such that
$\bar\alpha\triangleq\alpha-k\in(0,1)$ and write $u_t=kx_t+v_t$. Then
\[
x_{t+1}=\bar\alpha x_t-v_t+w_t.
\]
An SDAC filter for $v_t$ uses the powers $\bar\alpha^i$. If
$0\le v_t\le\bar\alpha x_t$, then
\[
kx_t\le u_t=kx_t+v_t
\le(k+\bar\alpha)x_t
=\alpha x_t,
\]
so the original feasibility constraint \eqref{eq:model-feasibility} is satisfied.
Since
\[
x_{t+1}\le\bar\alpha x_t+1,
\qquad x_1=0,
\]
induction gives $x_t\le1/(1-\bar\alpha)$. The transformed cost
$c_t(x,kx+v)$ remains convex in $(x,v)$. Thus, the analysis applies when the costs
satisfy the corresponding Lipschitz and bounded-subgradient assumptions on the
bounded domain
\[
0\le x\le\frac{1}{1-\bar\alpha},
\qquad
kx\le u\le\alpha x.
\]
For fixed $k$, the comparator class is restricted to policies
whose actions satisfy $u_t\ge kx_t$. Selecting $k$, either in advance or online,
is left for future work.

\paragraph{Vector actions and multi-task allocation}
Consider a single storage state $x_t$ with a vector action
$\mathbf{u}_t\in\mathbb{R}_{\ge0}^n$ that allocates resources across $n$ tasks,
subject to $\sum_{a=1}^n u_{t,a}\le\alpha x_t$. The dynamics are
$x_{t+1}=\alpha x_t-\sum_{a=1}^n u_{t,a}+w_t$, and the cost
$c_t(x,\mathbf{u})$ is convex in $(x,\mathbf{u})$. Replace
$\mathbf{m}\in\mathcal{M}_H$ with a matrix
$\mathbf{M}\in\mathbb{R}_{\ge0}^{n\times H}$ satisfying
\[
\sum_{a=1}^n\sum_{i=1}^H M_{a,i}\le1,
\qquad
u_{t,a}(\mathbf{M})
=\sum_{i=1}^H M_{a,i}\alpha^i w_{t-i}.
\]
Set $m_i=\sum_{a=1}^n M_{a,i}$. Then
\[
m_i\ge0,
\qquad
\sum_{i=1}^H m_i
=\sum_{a=1}^n\sum_{i=1}^H M_{a,i}
\le1,
\]
so $\mathbf{m}\in\mathcal{M}_H$. Moreover,
\[
\sum_{a=1}^n u_{t,a}(\mathbf{M})
=\sum_{i=1}^H m_i\alpha^i w_{t-i}
\le\alpha x_t(\mathbf{M}),
\]
where the inequality is the scalar SDAC feasibility bound.
The state and vector action are affine in $\mathbf{M}$, so
$c_t\bigl(x_t(\mathbf{M}),\mathbf{u}_t(\mathbf{M})\bigr)$ is convex in $\mathbf{M}$. For the
online SDAC controller, a nominal vector whose entries sum to more than
$\alpha x_t$ can be scaled proportionally so that its entries remain nonnegative
and sum to $\alpha x_t$. Under an $\ell_1$ Lipschitz bound in the action, the size of this
correction equals the excess total action, and the stability proof proceeds with
the corresponding matrix norm. OMD then runs over the matrix subprobability
simplex above.

\paragraph{Multiple storage states with diagonal dynamics}
Suppose the state is $\mathbf{x}_t\in\mathbb{R}_{\ge0}^d$ and evolves according to
$\mathbf{x}_{t+1}=\mathbf{A}\mathbf{x}_t-\mathbf{u}_t+\mathbf{w}_t$, where
$\mathbf{A}=\operatorname{diag}(\alpha_1,\ldots,\alpha_d)$,
$\alpha_s\in(0,1)$ for $s\in[d]$, the arrivals satisfy
$\mathbf{w}_t\in[0,1]^d$, and the per-storage constraints are
$0\le u_{t,s}\le\alpha_s x_{t,s}$. Assign one SDAC parameter in $\mathcal{M}_H$
to each storage state and use the powers $\alpha_s^i$ in coordinate $s$. The dynamics
and feasibility constraints then decouple by storage state, and the SDAC parameter domain
is a product of subprobability simplices, with
$x_{t,s}\le1/(1-\alpha_s)$ for every $s\in[d]$. The online SDAC controller clips
each nominal action to $\alpha_s x_{t,s}$. If the cost is jointly convex and
satisfies compatible Lipschitz and bounded-subgradient assumptions on the resulting bounded product
domain, the surrogate cost is jointly convex in all SDAC parameters and OMD applies
on the product set. With a sum of negentropy regularizers, the OMD minimization
separates by storage state, although each subgradient may depend on all storage states
and actions. Cross-coupled storage systems (nondiagonal $\mathbf{A}$) require
further analysis.

\section{Conclusion}
\label{sec:conclusion}

We studied adversarial online control of a storage system under a state-dependent
action constraint. We introduced SDAC, whose fixed policies are feasible by
construction, and developed an online SDAC controller with
$O\!\left(\sqrt{T\log(H+1)}\right)$ regret against the best fixed SDAC policy.
For $T\ge2t_{\mathrm{mix}}(\alpha)+2$, the
minimax lower bound gives
$\Omega((1-\alpha)^{-1}\sqrt{T})$ regret as $\alpha\uparrow1$.
A factor of $(1-\alpha)^{-1/2}$ remains between the upper and lower bounds
in their dependence on $1-\alpha$. The
SDAC approximation results further give
$O\!\left(\sqrt{T\log\log T}\right)$ regret against the best fixed
infinite-memory SDAC policy for a suitable choice $H=\Theta(\log T)$. The model applies directly to
energy-harvesting battery management and related settings in which storage limits
must be respected in real time.

\appendices
\section{Proofs for Simplex Disturbance-Action Control}
\label{sec:appendix-sdac-proofs}

\begin{proof}[Proof of Lemma~\ref{lem:sdac-state-decomp}]
For $1\le t<T$ and $i\in[H]$, the definitions give
\begin{align*}
r_{t+1}
&=\sum_{j=1}^{t}\alpha^{j-1}w_{t+1-j}
=w_t+\sum_{j=1}^{t-1}\alpha^j w_{t-j}
=w_t+\alpha r_t, \\
r_{t+1}(i)
&=\sum_{j=1}^{i}\alpha^{j-1}w_{t+1-j} \\
&=w_t+\sum_{j=1}^{i-1}\alpha^j w_{t-j}
=w_t+\alpha r_t(i)-\alpha^i w_{t-i}.
\end{align*}
The sums with upper limits $t-1$ and $i-1$ are empty when $t=1$ and $i=1$,
respectively.

At $t=1$, \eqref{eq:model-init} gives $x_1(\mathbf{m})=0$, while
$r_1(i)=r_1=0$ by the arrival and empty-sum conventions. Thus,
\eqref{eq:sdac-state-decomp} holds at $t=1$. Suppose it holds for some $1\le t<T$. By \eqref{eq:model-dynamics} and
\eqref{eq:def-sdac-policy},
\begin{align*}
x_{t+1}(\mathbf{m})
&=\alpha x_t(\mathbf{m})-u_t(\mathbf{m})+w_t \\
&=\alpha\left[
\sum_{i=1}^H m_i r_t(i)
+\left(1-\sum_{i=1}^H m_i\right)r_t
\right] \\
&\quad
-\sum_{i=1}^H m_i\alpha^i w_{t-i}+w_t \\
&=\sum_{i=1}^H m_i
\bigl(\alpha r_t(i)-\alpha^i w_{t-i}\bigr) \\
&\quad
+\left(1-\sum_{i=1}^H m_i\right)\alpha r_t+w_t \\
&=\sum_{i=1}^H m_i\bigl(r_{t+1}(i)-w_t\bigr) \\
&\quad
+\left(1-\sum_{i=1}^H m_i\right)(r_{t+1}-w_t)+w_t \\
&=\sum_{i=1}^H m_i r_{t+1}(i)
+\left(1-\sum_{i=1}^H m_i\right)r_{t+1}.
\end{align*}
Thus, \eqref{eq:sdac-state-decomp} follows by induction.
\end{proof}

\begin{proof}[Proof of Lemma~\ref{lem:sdac-surrogate}]
Lemma~\ref{lem:sdac-state-decomp} and \eqref{eq:def-sdac-policy} give
\begin{equation}
\label{eq:surr-affine-pair}
x_t(\mathbf{m})=r_t+\langle \boldsymbol{\phi}_t^x,\mathbf{m}\rangle,
\qquad
u_t(\mathbf{m})=\langle \boldsymbol{\phi}_t^u,\mathbf{m}\rangle.
\end{equation}
For $i\in[H]$, Lemma~\ref{lem:sdac-state-decomp} and the definitions of
$r_t(i)$ and $r_t$ give
\[
\begin{aligned}
|\phi_{t,i}^x|&=r_t-r_t(i)
=\sum_{j=i+1}^{t-1}\alpha^{j-1}w_{t-j} \\
&\le\sum_{j=i+1}^{\infty}\alpha^{j-1}
=\frac{\alpha^i}{1-\alpha}
\le\frac{\alpha}{1-\alpha}, \\
0\le\phi_{t,i}^u&=\alpha^i w_{t-i}\le\alpha^i\le\alpha.
\end{aligned}
\]
The sum is empty when $i\ge t-1$. Taking the maximum over $i\in[H]$ gives
\begin{equation}
\label{eq:surr-coefficient-norms}
\|\boldsymbol{\phi}_t^x\|_\infty\le\frac{\alpha}{1-\alpha},
\qquad
\|\boldsymbol{\phi}_t^u\|_\infty\le\alpha.
\end{equation}
Since the state--action pair in \eqref{eq:surr-affine-pair} is affine in
$\mathbf{m}$ and $c_t$ is convex, $\ell_t$ is convex.

By Lemma~\ref{lem:sdac-feasibility}, the fixed-policy state--action pairs
associated with $\mathbf{m}$ and $\mathbf{m}'$ belong to $\mathcal{D}_\alpha$.
Thus, \eqref{eq:model-lipschitz}, \eqref{eq:surr-affine-pair},
H\"older's inequality, and \eqref{eq:surr-coefficient-norms} give
\[
\begin{aligned}
|\ell_t(\mathbf{m})-\ell_t(\mathbf{m}')|
&\le
L_x|x_t(\mathbf{m})-x_t(\mathbf{m}')| \\
&\quad{}
+L_u|u_t(\mathbf{m})-u_t(\mathbf{m}')| \\
&=
L_x\bigl|\langle \boldsymbol{\phi}_t^x,\mathbf{m}-\mathbf{m}'\rangle\bigr| \\
&\quad{}
+
L_u\bigl|\langle \boldsymbol{\phi}_t^u,\mathbf{m}-\mathbf{m}'\rangle\bigr| \\
&\le
\left(\frac{\alpha L_x}{1-\alpha}+\alpha L_u\right)
\|\mathbf{m}-\mathbf{m}'\|_1 \\
&=
G_\alpha\|\mathbf{m}-\mathbf{m}'\|_1.
\end{aligned}
\]

Fix $\mathbf{m}\in\mathcal{M}_H$, and let $(\zeta_t^x,\zeta_t^u)$ and
$\mathbf{g}_t$ be defined by \eqref{eq:surr-subgradient}. For every
$\mathbf{m}'\in\mathcal{M}_H$, the subgradient inequality gives
\begin{align*}
\ell_t(\mathbf{m}')
&=c_t\bigl(x_t(\mathbf{m}'),u_t(\mathbf{m}')\bigr)
\ge c_t\bigl(x_t(\mathbf{m}),u_t(\mathbf{m})\bigr) \\
&\quad
+\zeta_t^x\bigl(x_t(\mathbf{m}')-x_t(\mathbf{m})\bigr)
+\zeta_t^u\bigl(u_t(\mathbf{m}')-u_t(\mathbf{m})\bigr) \\
&=\ell_t(\mathbf{m})
+\zeta_t^x\langle \boldsymbol{\phi}_t^x,\mathbf{m}'-\mathbf{m}\rangle
+\zeta_t^u\langle \boldsymbol{\phi}_t^u,\mathbf{m}'-\mathbf{m}\rangle \\
&=
\ell_t(\mathbf{m})+\langle \mathbf{g}_t,\mathbf{m}'-\mathbf{m}\rangle.
\end{align*}
Thus, $\mathbf{g}_t\in\partial\ell_t(\mathbf{m})$. The bounded-subgradient
assumption and \eqref{eq:surr-coefficient-norms} give
\begin{align*}
\|\mathbf{g}_t\|_\infty
&\le
|\zeta_t^x|\|\boldsymbol{\phi}_t^x\|_\infty
+|\zeta_t^u|\|\boldsymbol{\phi}_t^u\|_\infty
\le
\frac{\alpha L_x}{1-\alpha}+\alpha L_u \\
&=G_\alpha.\qedhere
\end{align*}
\end{proof}

\begin{proof}[Proof of Lemma~\ref{lem:sdac-truncation-approx}]
Since $\mathbf{m}\in\mathcal{M}_\infty$, we have $m_i\ge0$ for every $i\in[H]$ and
\[
\|\mathbf{m}_H(\mathbf{m})\|_1
=\sum_{i=1}^H m_i\le\sum_{i=1}^{\infty}m_i\le1.
\]
Thus, $\mathbf{m}_H(\mathbf{m})\in\mathcal{M}_H$.

To prove feasibility of the infinite-memory SDAC policy, truncate at length $T$.
Since $w_\tau=0$ for $\tau\le0$, for every $t\in[T]$,
\[
\begin{aligned}
u_t\bigl(\mathbf{m}_T(\mathbf{m})\bigr)
&=\sum_{i=1}^T m_i\alpha^i w_{t-i}\\
&=\sum_{i=1}^{t-1}m_i\alpha^i w_{t-i}=u_t(\mathbf{m}).
\end{aligned}
\]
The infinite-memory SDAC policy and its truncation have the same initial state
and actions, so their state trajectories coincide.
Lemma~\ref{lem:sdac-feasibility} therefore gives
\[
0\le u_t(\mathbf{m})\le\alpha x_t(\mathbf{m}),
\qquad t\in[T].
\]

For the approximation error, define
\[
\begin{aligned}
e_t^u&\triangleq
u_t(\mathbf{m})-u_t\bigl(\mathbf{m}_H(\mathbf{m})\bigr),\\
e_t^x&\triangleq
x_t\bigl(\mathbf{m}_H(\mathbf{m})\bigr)-x_t(\mathbf{m}).
\end{aligned}
\]
Again using $w_\tau=0$ for $\tau\le0$,
\[
e_t^u=\sum_{i=H+1}^{\infty}\alpha^i m_i w_{t-i},
\qquad
0\le e_t^u\le\sum_{i=H+1}^{\infty}\alpha^i m_i.
\]
Subtracting the two state recursions gives
\[
e_{t+1}^x=\alpha e_t^x+e_t^u,
\qquad e_1^x=0.
\]
Iterating this recursion gives, with an empty sum at $t=1$,
\[
\begin{aligned}
0\le e_t^x
&=\sum_{s=1}^{t-1}\alpha^{t-1-s}e_s^u\\
&\le\left(\sum_{i=H+1}^{\infty}\alpha^i m_i\right)
\sum_{s=1}^{t-1}\alpha^{t-1-s}\\
&=\left(\sum_{i=H+1}^{\infty}\alpha^i m_i\right)
\frac{1-\alpha^{t-1}}{1-\alpha}\\
&\le\frac{1}{1-\alpha}\sum_{i=H+1}^{\infty}\alpha^i m_i.
\end{aligned}
\]

Both policies are feasible, so their state--action pairs belong to
$\mathcal{D}_\alpha$. By \eqref{eq:model-lipschitz} and the triangle inequality,
\[
\begin{aligned}
&\left|J_T\bigl(\mathbf{m}_H(\mathbf{m})\bigr)-J_T(\mathbf{m})\right|\\
&\quad=\left|\sum_{t=1}^T\left[
c_t\bigl(x_t(\mathbf{m}_H(\mathbf{m})),u_t(\mathbf{m}_H(\mathbf{m}))\bigr)
\right.\right.\\
&\hspace{7em}\left.\left.
-c_t\bigl(x_t(\mathbf{m}),u_t(\mathbf{m})\bigr)\right]\right|\\
&\quad\le\sum_{t=1}^T\left(L_xe_t^x+L_ue_t^u\right)\\
&\quad\le\left(L_u+\frac{L_x}{1-\alpha}\right)T
\sum_{i=H+1}^{\infty}\alpha^i m_i.
\end{aligned}
\]
Finally, since $\alpha\in(0,1)$,
\[
\begin{aligned}
\sum_{i=H+1}^{\infty}\alpha^i m_i
&\le\alpha^{H+1}\sum_{i=H+1}^{\infty}m_i\\
&\le\alpha^{H+1}\sum_{i=1}^{\infty}m_i
\le\alpha^{H+1}.
\end{aligned}
\]
This proves \eqref{eq:sdac-truncation-cost-gap}.
\end{proof}

\begin{proof}[Proof of Corollary~\ref{cor:fraction-as-sdac}]
Since $0\le k/\alpha\le1$,
\[
m_i(k)=\frac{k}{\alpha}\left(1-\frac{k}{\alpha}\right)^{i-1}\ge0,
\qquad i\ge1.
\]
If $k=0$, then $\mathbf{m}(k)=\mathbf{0}\in\mathcal{M}_\infty$.
If $k>0$, then $0\le1-k/\alpha<1$ and
\[
\begin{aligned}
\sum_{i=1}^{\infty}m_i(k)
&=\frac{k}{\alpha}\sum_{j=0}^{\infty}
\left(1-\frac{k}{\alpha}\right)^j\\
&=\frac{k}{\alpha}\frac{1}{1-(1-k/\alpha)}
=\frac{k}{\alpha}\frac{\alpha}{k}=1.
\end{aligned}
\]
Thus, $\mathbf{m}(k)\in\mathcal{M}_\infty$ for every $k\in[0,\alpha]$.

Under the fixed-fraction policy, \eqref{eq:model-dynamics} becomes
\[
x_{t+1}(k)=(\alpha-k)x_t(k)+w_t.
\]
Iterating this recursion from $x_1(k)=0$ gives
\[
x_t(k)=\sum_{i=1}^{t-1}(\alpha-k)^{i-1}w_{t-i}.
\]
For every $i\ge1$,
\[
\begin{aligned}
\alpha^i m_i(k)
&=\alpha^i\frac{k}{\alpha}
\left(\frac{\alpha-k}{\alpha}\right)^{i-1}\\
&=k\alpha^{i-1}\frac{(\alpha-k)^{i-1}}{\alpha^{i-1}}
=k(\alpha-k)^{i-1}.
\end{aligned}
\]
Therefore,
\[
\begin{aligned}
u_t(k)&=kx_t(k)
=\sum_{i=1}^{t-1}k(\alpha-k)^{i-1}w_{t-i}\\
&=\sum_{i=1}^{t-1}\alpha^i m_i(k)w_{t-i}
=u_t\bigl(\mathbf{m}(k)\bigr).
\end{aligned}
\]
Both state trajectories start at zero and follow \eqref{eq:model-dynamics} with
the same actions. Hence, $x_t\bigl(\mathbf{m}(k)\bigr)=x_t(k)$ for every $t\in[T]$.
\end{proof}

\section{Proofs for the Online SDAC Controller}
\label{sec:appendix-algorithm-proofs}

The OMD results for unnormalized negentropy and its Bregman divergence on the
probability simplex are standard \cite{hazan2016introduction}. We prove the
corresponding forms for the subprobability simplex $\mathcal{M}_H$ used below.

\begin{lemma}[Strong convexity of negentropy on the subprobability simplex]
\label{lem:subprob-negentropy-strong-convexity}
For every $\mathbf{m}\in\mathcal{M}_H$ and
$\mathbf{v}\in\mathcal{M}_H\cap\mathbb{R}_{>0}^H$,
\begin{equation}
\label{eq:subprob-negentropy-strong-convexity}
D_\Psi(\mathbf{m}\|\mathbf{v})
\ge
\frac12\|\mathbf{m}-\mathbf{v}\|_1^2.
\end{equation}
Thus, $\Psi$ is $1$-strongly convex with respect to $\|\cdot\|_1$ on
$\mathcal{M}_H$.
\end{lemma}

\begin{proof}
First suppose $\mathbf{m}>0$ coordinatewise. Set
$\mathbf{d}\triangleq\mathbf{m}-\mathbf{v}$ and
$\mathbf{z}_s\triangleq\mathbf{v}+s\mathbf{d}$ for $s\in[0,1]$. Then
$\mathbf{z}_s\in\mathcal{M}_H\cap\mathbb{R}_{>0}^H$, and weighted
Cauchy--Schwarz gives
\[
\begin{aligned}
\|\mathbf{d}\|_1^2
&=\left(\sum_{i=1}^H
\frac{|d_i|}{\sqrt{z_{s,i}}}\sqrt{z_{s,i}}\right)^2
\le
\left(\sum_{i=1}^H\frac{d_i^2}{z_{s,i}}\right)
\left(\sum_{i=1}^H z_{s,i}\right) \\
&\le
\mathbf{d}^\top\nabla^2\Psi(\mathbf{z}_s)\mathbf{d}.
\end{aligned}
\]
Taylor's formula with integral remainder gives
\[
\begin{aligned}
D_\Psi(\mathbf{m}\|\mathbf{v})
&=\int_0^1(1-s)
\mathbf{d}^\top\nabla^2\Psi(\mathbf{z}_s)\mathbf{d}\,ds \\
&\ge
\|\mathbf{d}\|_1^2\int_0^1(1-s)\,ds
=\frac12\|\mathbf{d}\|_1^2.
\end{aligned}
\]
For a general $\mathbf{m}\in\mathcal{M}_H$ and $0<\varepsilon<1$, apply the bound to
$\mathbf{m}^{\varepsilon}\triangleq(1-\varepsilon)\mathbf{m}
+\varepsilon\mathbf{v}$ and let $\varepsilon\downarrow0$. The result follows
from the convention $0\log0=0$ and continuity in the first argument of
$D_\Psi$.
\end{proof}

\begin{lemma}[Action clipping as a Bregman projection]
\label{lem:availability-projection}
For $t\in[T]$, define the set of parameters compatible with the current state by
\[
\mathcal{B}_t
\triangleq
\left\{
\mathbf{m}\in\mathcal{M}_H:
\langle\boldsymbol{\phi}_t^u,\mathbf{m}\rangle
\le \alpha x_t
\right\}.
\]
Let $\mathbf{m}_t^{\mathrm{exec}}$ be the Bregman projection of
$\mathbf{m}_t$ onto $\mathcal{B}_t$:
\[
\mathbf{m}_t^{\mathrm{exec}}
\in
\operatorname*{arg\,min}_{\mathbf{m}\in\mathcal{B}_t}
D_\Psi(\mathbf{m}\|\mathbf{m}_t).
\]
Then
\[
\left\langle
\boldsymbol{\phi}_t^u,\mathbf{m}_t^{\mathrm{exec}}
\right\rangle
=
\min\!\left\{
\left\langle\boldsymbol{\phi}_t^u,\mathbf{m}_t\right\rangle,
\alpha x_t
\right\}
=u_t.
\]
\end{lemma}

\begin{proof}
The online state satisfies $x_t\ge0$, so
$\mathbf{0}\in\mathcal{B}_t$. This set is compact and convex because it is
the intersection of $\mathcal{M}_H$ and a closed halfspace. The initialization
\eqref{eq:alg-init} and update
\eqref{eq:alg-omd-mult}--\eqref{eq:alg-omd-proj} ensure that every coordinate of
$\mathbf{m}_t$ is positive. By \eqref{eq:def-bregman},
$D_\Psi(\cdot\|\mathbf{m}_t)$ is continuous on $\mathcal{M}_H$. By
Lemma~\ref{lem:subprob-negentropy-strong-convexity}, it is strictly convex and,
for every $\mathbf{m}\in\mathcal{M}_H$,
\[
D_\Psi(\mathbf{m}\|\mathbf{m}_t)
\ge \frac12\|\mathbf{m}-\mathbf{m}_t\|_1^2
\ge0.
\]
Thus, the projection exists and is unique, and equality holds only at
$\mathbf{m}=\mathbf{m}_t$.

If $\langle\boldsymbol{\phi}_t^u,\mathbf{m}_t\rangle\le\alpha x_t$, then
$\mathbf{m}_t\in\mathcal{B}_t$ and
$\mathbf{m}_t^{\mathrm{exec}}=\mathbf{m}_t$.

Now suppose $\langle\boldsymbol{\phi}_t^u,\mathbf{m}_t\rangle>\alpha x_t$.
Since $\mathbf{m}_t^{\mathrm{exec}}\in\mathcal{B}_t$,
\[
\langle\boldsymbol{\phi}_t^u,\mathbf{m}_t^{\mathrm{exec}}\rangle
\le\alpha x_t
<\langle\boldsymbol{\phi}_t^u,\mathbf{m}_t\rangle.
\]
Assume, for contradiction, that the action constraint is slack:
\[
\langle\boldsymbol{\phi}_t^u,\mathbf{m}_t^{\mathrm{exec}}\rangle
<\alpha x_t
<\langle\boldsymbol{\phi}_t^u,\mathbf{m}_t\rangle.
\]
Choose $0<\delta<1$ such that
\[
\delta\left(
\langle\boldsymbol{\phi}_t^u,\mathbf{m}_t\rangle-
\langle\boldsymbol{\phi}_t^u,\mathbf{m}_t^{\mathrm{exec}}\rangle
\right)
<
\alpha x_t-
\langle\boldsymbol{\phi}_t^u,\mathbf{m}_t^{\mathrm{exec}}\rangle.
\]
Such a $\delta$ exists because both differences are positive. Set
\[
\mathbf{m}_\delta
\triangleq
(1-\delta)\mathbf{m}_t^{\mathrm{exec}}+\delta\mathbf{m}_t.
\]
By convexity, $\mathbf{m}_\delta\in\mathcal{M}_H$, and
\[
\begin{aligned}
&\langle\boldsymbol{\phi}_t^u,\mathbf{m}_\delta\rangle \\
&=
\langle\boldsymbol{\phi}_t^u,\mathbf{m}_t^{\mathrm{exec}}\rangle
+\delta\left(
\langle\boldsymbol{\phi}_t^u,\mathbf{m}_t\rangle-
\langle\boldsymbol{\phi}_t^u,\mathbf{m}_t^{\mathrm{exec}}\rangle
\right) <\alpha x_t.
\end{aligned}
\]
Thus, $\mathbf{m}_\delta\in\mathcal{B}_t$. Since
$\mathbf{m}_t\notin\mathcal{B}_t$,
$\mathbf{m}_t^{\mathrm{exec}}\ne\mathbf{m}_t$, and hence
$D_\Psi(\mathbf{m}_t^{\mathrm{exec}}\|\mathbf{m}_t)>0$. Using strict convexity
and $D_\Psi(\mathbf{m}_t\|\mathbf{m}_t)=0$, we obtain
\begin{align*}
D_\Psi(\mathbf{m}_\delta\|\mathbf{m}_t)
&<
(1-\delta)D_\Psi(\mathbf{m}_t^{\mathrm{exec}}\|\mathbf{m}_t)
+\delta D_\Psi(\mathbf{m}_t\|\mathbf{m}_t) \\
&=
(1-\delta)D_\Psi(\mathbf{m}_t^{\mathrm{exec}}\|\mathbf{m}_t)
<
D_\Psi(\mathbf{m}_t^{\mathrm{exec}}\|\mathbf{m}_t),
\end{align*}
contradicting the definition of $\mathbf{m}_t^{\mathrm{exec}}$. Therefore, the
constraint is active:
\[
\langle\boldsymbol{\phi}_t^u,\mathbf{m}_t^{\mathrm{exec}}\rangle
=\alpha x_t.
\]
Combining the two cases with \eqref{eq:alg-feasible-action} proves the result.
\end{proof}

\begin{lemma}[Closed-form solution of the OMD update]
\label{lem:omd-closed-form}
Fix $\mathbf{m}_t\in\mathcal{M}_H\cap\mathbb{R}_{>0}^H$, $\eta>0$, and
$\mathbf{g}_t\in\mathbb{R}^H$. Define
$\widetilde{m}_{t+1,i}\triangleq m_{t,i}\exp(-\eta g_{t,i})$ for $i\in[H]$.
Then the OMD update \eqref{eq:alg-omd-argmin} has the unique solution
\[
\mathbf{m}_{t+1}
=\frac{\widetilde{\mathbf{m}}_{t+1}}
{\max\{1,\|\widetilde{\mathbf{m}}_{t+1}\|_1\}}.
\]
Moreover, $\mathbf{m}_{t+1}\in\mathcal{M}_H\cap\mathbb{R}_{>0}^H$. Thus, the
OMD iterates initialized by \eqref{eq:alg-init} remain strictly positive.
\end{lemma}

\begin{proof}
Since $m_{t,i}>0$, we have $\widetilde{m}_{t+1,i}>0$ for every $i\in[H]$.
For any $\mathbf{m}\in\mathcal{M}_H$,
\[
m_i\log\frac{m_i}{\widetilde{m}_{t+1,i}}
=m_i\log\frac{m_i}{m_{t,i}}+\eta g_{t,i}m_i,
\]
where the identity at $m_i=0$ uses the convention $0\log0=0$. Substituting into
\eqref{eq:def-bregman} gives
\[
\begin{aligned}
&\eta\langle\mathbf{g}_t,\mathbf{m}\rangle+D_\Psi(\mathbf{m}\|\mathbf{m}_t) \\
&=\sum_{i=1}^H\left(
\eta g_{t,i}m_i+m_i\log\frac{m_i}{m_{t,i}}-m_i+m_{t,i}
\right) \\
&=\sum_{i=1}^H\left(
m_i\log\frac{m_i}{\widetilde{m}_{t+1,i}}-m_i+m_{t,i}
\right) \\
&=D_\Psi(\mathbf{m}\|\widetilde{\mathbf{m}}_{t+1})
+\sum_{i=1}^H(m_{t,i}-\widetilde{m}_{t+1,i}).
\end{aligned}
\]
Thus, it suffices to minimize
$D_\Psi(\mathbf{m}\|\widetilde{\mathbf{m}}_{t+1})$ over $\mathcal{M}_H$.

If $\|\widetilde{\mathbf{m}}_{t+1}\|_1\le1$, then
$\widetilde{\mathbf{m}}_{t+1}\in\mathcal{M}_H$.
Lemma~\ref{lem:subprob-negentropy-strong-convexity} gives
\[
D_\Psi(\mathbf{m}\|\widetilde{\mathbf{m}}_{t+1})
\ge\frac12\|\mathbf{m}-\widetilde{\mathbf{m}}_{t+1}\|_1^2\ge0.
\]
Since $D_\Psi(\widetilde{\mathbf{m}}_{t+1}\|\widetilde{\mathbf{m}}_{t+1})=0$,
$\widetilde{\mathbf{m}}_{t+1}$ is the unique minimizer.

If $\|\widetilde{\mathbf{m}}_{t+1}\|_1>1$, set
\[
\mathbf{m}^\star
\triangleq
\frac{\widetilde{\mathbf{m}}_{t+1}}
{\|\widetilde{\mathbf{m}}_{t+1}\|_1}.
\]
Then $\mathbf{m}^\star>0$ coordinatewise, $\|\mathbf{m}^\star\|_1=1$, and
\[
\widetilde{m}_{t+1,i}
=\|\widetilde{\mathbf{m}}_{t+1}\|_1m_i^\star.
\]
Thus, for every $i\in[H]$,
\[
\begin{aligned}
m_i\log\frac{m_i}{\widetilde{m}_{t+1,i}}
&=m_i\log\frac{m_i}{m_i^\star}
-m_i\log\|\widetilde{\mathbf{m}}_{t+1}\|_1, \\
m_i^\star\log\frac{m_i^\star}{\widetilde{m}_{t+1,i}}
&=-m_i^\star\log\|\widetilde{\mathbf{m}}_{t+1}\|_1.
\end{aligned}
\]
Using these identities in \eqref{eq:def-bregman} gives, for every
$\mathbf{m}\in\mathcal{M}_H$,
\[
\begin{aligned}
&D_\Psi(\mathbf{m}\|\widetilde{\mathbf{m}}_{t+1})
-D_\Psi(\mathbf{m}^\star\|\widetilde{\mathbf{m}}_{t+1}) \\
&=\sum_{i=1}^H
\left(
m_i\log\frac{m_i}{m_i^\star}-m_i+m_i^\star
\right)
\\
&\quad+\sum_{i=1}^H(m_i^\star-m_i)
\log\|\widetilde{\mathbf{m}}_{t+1}\|_1 \\
&=D_\Psi(\mathbf{m}\|\mathbf{m}^\star)
+\bigl(1-\|\mathbf{m}\|_1\bigr)
\log\|\widetilde{\mathbf{m}}_{t+1}\|_1
\ge0.
\end{aligned}
\]
Since $\|\mathbf{m}\|_1\le1$ and $\log\|\widetilde{\mathbf{m}}_{t+1}\|_1>0$,
Lemma~\ref{lem:subprob-negentropy-strong-convexity} gives the last inequality,
with equality only at $\mathbf{m}=\mathbf{m}^\star$. Thus,
$\mathbf{m}^\star$ is the unique minimizer.

In both cases, the minimizer has the stated form and is strictly positive.
Induction from \eqref{eq:alg-init} proves that every iterate is strictly positive.
\end{proof}

\section{Proofs for the Regret Upper Bounds}
\label{sec:appendix-regret-proofs}

\begin{lemma}[One-step OMD movement bound]
\label{lem:omd-move}
The OMD update \eqref{eq:alg-omd-argmin} with step size $\eta>0$ satisfies
\[
\|\mathbf{m}_{t+1}-\mathbf{m}_t\|_1
\le \eta\|\mathbf{g}_t\|_\infty
\le \eta G_\alpha.
\]
\end{lemma}

\begin{proof}
By Lemma~\ref{lem:omd-closed-form}, both iterates are positive. Applying
Lemma~\ref{lem:subprob-negentropy-strong-convexity} in both directions gives
\[
\begin{aligned}
D_\Psi(\mathbf{m}_{t+1}\|\mathbf{m}_t)
&\ge
\frac12\|\mathbf{m}_{t+1}-\mathbf{m}_t\|_1^2, \\
D_\Psi(\mathbf{m}_t\|\mathbf{m}_{t+1})
&\ge
\frac12\|\mathbf{m}_{t+1}-\mathbf{m}_t\|_1^2.
\end{aligned}
\]
Adding these bounds and using the symmetric Bregman identity gives
\[
\begin{aligned}
&\|\mathbf{m}_{t+1}-\mathbf{m}_t\|_1^2 \\
&\quad\le
D_\Psi(\mathbf{m}_{t+1}\|\mathbf{m}_t)+D_\Psi(\mathbf{m}_t\|\mathbf{m}_{t+1}) \\
&\quad=
\bigl\langle\nabla\Psi(\mathbf{m}_{t+1})-\nabla\Psi(\mathbf{m}_t),
\mathbf{m}_{t+1}-\mathbf{m}_t\bigr\rangle \\
&\quad\le
\eta\langle\mathbf{g}_t,\mathbf{m}_t-\mathbf{m}_{t+1}\rangle \\
&\quad\le
\eta\|\mathbf{g}_t\|_\infty\|\mathbf{m}_{t+1}-\mathbf{m}_t\|_1.
\end{aligned}
\]
The last two inequalities use the first-order optimality condition for
\eqref{eq:alg-omd-argmin} with comparison point $\mathbf{m}_t$ and
H\"older's inequality, respectively. Dividing by
$\|\mathbf{m}_{t+1}-\mathbf{m}_t\|_1$ gives the first bound when the iterates
differ; when they coincide, its left-hand side is zero.
Finally, $\|\mathbf{g}_t\|_\infty\le G_\alpha$ gives the uniform bound.
\end{proof}

\begin{lemma}[OMD regret for the surrogate costs]
\label{lem:omd-surrogate-regret}
For every $\mathbf{m}\in\mathcal{M}_H$, the update \eqref{eq:alg-omd-argmin} satisfies
\begin{equation}
\label{eq:omd-surrogate-regret}
\sum_{t=1}^T\ell_t(\mathbf{m}_t)
-\sum_{t=1}^T\ell_t(\mathbf{m})
\le
\frac{D_\Psi(\mathbf{m}\|\mathbf{m}_1)}{\eta}
+
\frac{\eta}{2}\sum_{t=1}^T\|\mathbf{g}_t\|_\infty^2.
\end{equation}
\end{lemma}

\begin{proof}
Convexity of $\ell_t$ gives
\begin{equation}
\label{eq:omd-convexity-step}
\begin{aligned}
\ell_t(\mathbf{m}_t)-\ell_t(\mathbf{m})
&\le
\langle \mathbf{g}_t,\mathbf{m}_t-\mathbf{m}\rangle
=
\langle \mathbf{g}_t,\mathbf{m}_t-\mathbf{m}_{t+1}\rangle \\
&\quad+
\langle \mathbf{g}_t,\mathbf{m}_{t+1}-\mathbf{m}\rangle.
\end{aligned}
\end{equation}
The first-order optimality condition for \eqref{eq:alg-omd-argmin}, with
comparison point $\mathbf{m}\in\mathcal{M}_H$, and the three-point Bregman
identity give
\begin{equation}
\label{eq:omd-three-point-step}
\begin{aligned}
&\eta\langle \mathbf{g}_t,\mathbf{m}_{t+1}-\mathbf{m}\rangle \\
&\quad\le
\langle \nabla\Psi(\mathbf{m}_{t+1})-\nabla\Psi(\mathbf{m}_t),\mathbf{m}-\mathbf{m}_{t+1}\rangle \\
&\quad=
D_\Psi(\mathbf{m}\|\mathbf{m}_t)
-D_\Psi(\mathbf{m}\|\mathbf{m}_{t+1}) \\
&\qquad\quad
-D_\Psi(\mathbf{m}_{t+1}\|\mathbf{m}_t).
\end{aligned}
\end{equation}
Substituting \eqref{eq:omd-three-point-step} into
\eqref{eq:omd-convexity-step} gives
\begin{equation}
\label{eq:omd-one-step-regret}
\begin{aligned}
\ell_t(\mathbf{m}_t)-\ell_t(\mathbf{m})
&\le
\frac{D_\Psi(\mathbf{m}\|\mathbf{m}_t)}{\eta}
-\frac{D_\Psi(\mathbf{m}\|\mathbf{m}_{t+1})}{\eta} \\
&\quad+
\langle \mathbf{g}_t,\mathbf{m}_t-\mathbf{m}_{t+1}\rangle
\\
&\quad
-\frac{D_\Psi(\mathbf{m}_{t+1}\|\mathbf{m}_t)}{\eta}.
\end{aligned}
\end{equation}
Lemma~\ref{lem:subprob-negentropy-strong-convexity} and H\"older's inequality give
\begin{equation}
\label{eq:omd-young-step}
\begin{aligned}
&\langle \mathbf{g}_t,\mathbf{m}_t-\mathbf{m}_{t+1}\rangle
-\frac{D_\Psi(\mathbf{m}_{t+1}\|\mathbf{m}_t)}{\eta} \\
&\qquad\le
\|\mathbf{g}_t\|_\infty\|\mathbf{m}_{t+1}-\mathbf{m}_t\|_1
-\frac{\|\mathbf{m}_{t+1}-\mathbf{m}_t\|_1^2}{2\eta} \\
&\qquad\le
\frac{\eta}{2}\|\mathbf{g}_t\|_\infty^2.
\end{aligned}
\end{equation}
The last step is Young's inequality. Substituting \eqref{eq:omd-young-step} into
\eqref{eq:omd-one-step-regret} and summing over $t$ gives
\[
\begin{aligned}
\sum_{t=1}^T\ell_t(\mathbf{m}_t)
-\sum_{t=1}^T\ell_t(\mathbf{m})
&\le
\frac{D_\Psi(\mathbf{m}\|\mathbf{m}_1)}{\eta}
-\frac{D_\Psi(\mathbf{m}\|\mathbf{m}_{T+1})}{\eta} \\
&\quad+
\frac{\eta}{2}\sum_{t=1}^T\|\mathbf{g}_t\|_\infty^2.
\end{aligned}
\]
Since $D_\Psi(\mathbf{m}\|\mathbf{m}_{T+1})\ge0$, this proves
\eqref{eq:omd-surrogate-regret}.
\end{proof}

\begin{lemma}[Positive-part representation of the dynamics]
\label{lem:positive-part-dynamics}
For $t\in[T]$, $x\ge0$, and $\mathbf{m}\in\mathcal{M}_H$, define
\begin{equation}
\label{eq:def-contractive-transition}
F_t(x;\mathbf{m})
\triangleq
\left[\alpha x-\langle\boldsymbol{\phi}_t^u,\mathbf{m}\rangle\right]_+
+w_t.
\end{equation}
Then
\begin{equation}
\label{eq:positive-part-transitions}
x_{t+1}=F_t(x_t;\mathbf{m}_t),
\qquad
x_{t+1}(\mathbf{m})=F_t(x_t(\mathbf{m});\mathbf{m}).
\end{equation}
Moreover, for all $x,x'\ge0$ and
$\mathbf{m},\mathbf{m}'\in\mathcal{M}_H$,
\begin{equation}
\label{eq:transition-contraction}
\begin{aligned}
&|F_t(x;\mathbf{m})-F_t(x';\mathbf{m}')| \\
&\qquad\le
\alpha|x-x'|
+\|\boldsymbol{\phi}_t^u\|_\infty
\|\mathbf{m}-\mathbf{m}'\|_1.
\end{aligned}
\end{equation}
\end{lemma}

\begin{proof}
By \eqref{eq:model-dynamics} and \eqref{eq:alg-feasible-action},
\[
\begin{aligned}
x_{t+1}
&=\alpha x_t-
\min\!\left\{
\langle\boldsymbol{\phi}_t^u,\mathbf{m}_t\rangle,
\alpha x_t
\right\}+w_t \\
&=\left[
\alpha x_t-\langle\boldsymbol{\phi}_t^u,\mathbf{m}_t\rangle
\right]_++w_t
=F_t(x_t;\mathbf{m}_t).
\end{aligned}
\]
For a fixed $\mathbf{m}$, \eqref{eq:def-sdac-policy} and
Lemma~\ref{lem:sdac-feasibility} give
$\langle\boldsymbol{\phi}_t^u,\mathbf{m}\rangle
=u_t(\mathbf{m})\le\alpha x_t(\mathbf{m})$. Hence, by
\eqref{eq:model-dynamics},
\[
\begin{aligned}
x_{t+1}(\mathbf{m})
&=\alpha x_t(\mathbf{m})
-\langle\boldsymbol{\phi}_t^u,\mathbf{m}\rangle+w_t \\
&=\left[
\alpha x_t(\mathbf{m})
-\langle\boldsymbol{\phi}_t^u,\mathbf{m}\rangle
\right]_++w_t
=F_t(x_t(\mathbf{m});\mathbf{m}).
\end{aligned}
\]
Finally, since $|[z]_+-[z']_+|\le|z-z'|$,
\[
\begin{aligned}
&|F_t(x;\mathbf{m})-F_t(x';\mathbf{m}')| \\
&=\left|
\left[\alpha x-\langle\boldsymbol{\phi}_t^u,\mathbf{m}\rangle\right]_+
-\left[\alpha x'-\langle\boldsymbol{\phi}_t^u,\mathbf{m}'\rangle\right]_+
\right| \\
&\le\left|
\alpha(x-x')
-\langle\boldsymbol{\phi}_t^u,\mathbf{m}-\mathbf{m}'\rangle
\right| \\
&\le\alpha|x-x'|
+\|\boldsymbol{\phi}_t^u\|_\infty
\|\mathbf{m}-\mathbf{m}'\|_1,
\end{aligned}
\]
where the last step uses the triangle inequality and H\"older's inequality.
\end{proof}

\begin{proof}[Proof of Lemma~\ref{lem:contractive-deviation}]
Let $d_t\triangleq|x_t-x_t(\mathbf{m}_t)|$. Since both trajectories start at
zero, $d_1=0$. For $1\le t<T$, Lemma~\ref{lem:positive-part-dynamics}, applied
with the same parameter $\mathbf{m}_t$, gives
\[
\begin{aligned}
|x_{t+1}-x_{t+1}(\mathbf{m}_t)|
&=|F_t(x_t;\mathbf{m}_t)-F_t(x_t(\mathbf{m}_t);\mathbf{m}_t)| \\
&\le\alpha d_t.
\end{aligned}
\]
The affine state formula \eqref{eq:surr-affine-pair}, the coefficient bound
\eqref{eq:surr-coefficient-norms}, and Lemma~\ref{lem:omd-move} give
\[
\begin{aligned}
&|x_{t+1}(\mathbf{m}_t)-x_{t+1}(\mathbf{m}_{t+1})| \\
&\qquad=|\langle\boldsymbol{\phi}_{t+1}^x,
\mathbf{m}_t-\mathbf{m}_{t+1}\rangle| \\
&\qquad\le\|\boldsymbol{\phi}_{t+1}^x\|_\infty
\|\mathbf{m}_{t+1}-\mathbf{m}_t\|_1 \\
&\qquad\le\frac{\alpha\eta G_\alpha}{1-\alpha}.
\end{aligned}
\]
Thus, the triangle inequality yields
\[
\begin{aligned}
d_{t+1}
&\le|x_{t+1}-x_{t+1}(\mathbf{m}_t)|
+|x_{t+1}(\mathbf{m}_t)-x_{t+1}(\mathbf{m}_{t+1})| \\
&\le\alpha d_t+\frac{\alpha\eta G_\alpha}{1-\alpha}.
\end{aligned}
\]
Iterating from $d_1=0$ gives, for $2\le t\le T$,
\[
\begin{aligned}
d_t
&\le\frac{\alpha\eta G_\alpha}{1-\alpha}
\sum_{j=0}^{t-2}\alpha^j
\le\frac{\alpha\eta G_\alpha}{1-\alpha}
\sum_{j=0}^{\infty}\alpha^j
=\frac{\alpha\eta G_\alpha}{(1-\alpha)^2}.
\end{aligned}
\]
This bound also holds at $t=1$.

The clipping rule \eqref{eq:alg-feasible-action} and feasibility of each fixed
policy (Lemma~\ref{lem:sdac-feasibility}) give
\[
\begin{aligned}
0\le u_t(\mathbf{m}_t)-u_t
&=[u_t(\mathbf{m}_t)-\alpha x_t]_+
\le[\alpha x_t(\mathbf{m}_t)-\alpha x_t]_+ \\
&\le\alpha d_t.
\end{aligned}
\]
Finally, $\ell_t(\mathbf{m}_t)=c_t(x_t(\mathbf{m}_t),u_t(\mathbf{m}_t))$.
Both state--action pairs lie in $\mathcal{D}_\alpha$, so the Lipschitz bound
\eqref{eq:model-lipschitz} gives
\[
\begin{aligned}
c_t(x_t,u_t)-\ell_t(\mathbf{m}_t)
&\le L_xd_t+L_u|u_t-u_t(\mathbf{m}_t)| \\
&\le(L_x+\alpha L_u)d_t
\le\frac{\alpha\eta G_\alpha(L_x+\alpha L_u)}{(1-\alpha)^2}.
\end{aligned}
\]
Summing over $t=1,\ldots,T$ proves \eqref{eq:stability-sum-bound}.
\end{proof}

\begin{proof}[Proof of Theorem~\ref{thm:regret-bound}]
Fix $\mathbf{m}\in\mathcal{M}_H$. We first bound the initial Bregman
divergence. Since $\mathbf{m}_1=\frac{1}{H+1}\mathbf{1}$ and
$m_i\log m_i\le0$ for $m_i\in[0,1]$,
\[
\begin{aligned}
D_\Psi(\mathbf{m}\|\mathbf{m}_1)
&=\sum_{i=1}^H m_i\log m_i \\
&\quad+\|\mathbf{m}\|_1\bigl(\log(H+1)-1\bigr)
+\frac{H}{H+1} \\
&\le (1-\|\mathbf{m}\|_1)\frac{H}{H+1} \\
&\quad+\|\mathbf{m}\|_1
\left(\log(H+1)-\frac{1}{H+1}\right) \\
&\le\log(H+1).
\end{aligned}
\]
The last inequality uses $\|\mathbf{m}\|_1\in[0,1]$ and
$\frac{H}{H+1}\le\log(H+1)$, which follows from
$\log z\ge1-z^{-1}$ for $z\ge1$.

Lemma~\ref{lem:omd-surrogate-regret} and
$\|\mathbf{g}_t\|_\infty\le G_\alpha$ now give
\[
\begin{aligned}
\sum_{t=1}^T
\bigl(\ell_t(\mathbf{m}_t)-\ell_t(\mathbf{m})\bigr)
&\le\frac{D_\Psi(\mathbf{m}\|\mathbf{m}_1)}{\eta}
+\frac{\eta}{2}\sum_{t=1}^T\|\mathbf{g}_t\|_\infty^2 \\
&\le\frac{\log(H+1)}{\eta}
+\frac{\eta T G_\alpha^2}{2}.
\end{aligned}
\]
Since $\ell_t(\mathbf{m})=c_t(x_t(\mathbf{m}),u_t(\mathbf{m}))$,
adding the stability bound \eqref{eq:stability-sum-bound} yields
\[
\begin{aligned}
J_T^{\mathrm{on}}-J_T(\mathbf{m})
&=\sum_{t=1}^T\bigl(c_t(x_t,u_t)-\ell_t(\mathbf{m}_t)\bigr) \\
&\quad+\sum_{t=1}^T
\bigl(\ell_t(\mathbf{m}_t)-\ell_t(\mathbf{m})\bigr) \\
&\le\frac{\log(H+1)}{\eta}
+\frac{\eta T G_\alpha^2}{2} \\
&\quad+\frac{\alpha\eta T G_\alpha(L_x+\alpha L_u)}
{(1-\alpha)^2} \\
&=\frac{\log(H+1)}{\eta}+\eta T\Gamma_\alpha.
\end{aligned}
\]
Taking the maximum over $\mathbf{m}\in\mathcal{M}_H$ proves
\eqref{eq:regret-theorem-bound}. At the stated step size,
\[
\frac{\log(H+1)}{\eta^*}
=\eta^*T\Gamma_\alpha
=\sqrt{\Gamma_\alpha T\log(H+1)},
\]
which gives the stated regret bound.
\end{proof}

\begin{proof}[Proof of Corollary~\ref{cor:regret-vs-infinite-sdac}]
Fix an infinite-memory SDAC policy $\mathbf{m}\in\mathcal{M}_\infty$.
By Lemma~\ref{lem:sdac-truncation-approx}, its truncation satisfies
$\mathbf{m}_{H_T}(\mathbf{m})\in\mathcal{M}_{H_T}$.
The definition of $\RegSDAC{T}(H_T)$ gives
\[
\begin{aligned}
J_T^{\mathrm{on}}-J_T(\mathbf{m})
&=J_T^{\mathrm{on}}-J_T\bigl(\mathbf{m}_{H_T}(\mathbf{m})\bigr)\\
&\quad+J_T\bigl(\mathbf{m}_{H_T}(\mathbf{m})\bigr)-J_T(\mathbf{m})\\
&\le\RegSDAC{T}(H_T)\\
&\quad+J_T\bigl(\mathbf{m}_{H_T}(\mathbf{m})\bigr)-J_T(\mathbf{m}).
\end{aligned}
\]
The approximation bound \eqref{eq:sdac-truncation-cost-gap} gives
\[
J_T\bigl(\mathbf{m}_{H_T}(\mathbf{m})\bigr)-J_T(\mathbf{m})
\le\left(L_u+\frac{L_x}{1-\alpha}\right)\alpha^{H_T+1}T.
\]
This bound is independent of $\mathbf{m}$. Taking the supremum and applying
Theorem~\ref{thm:regret-bound} with memory $H_T$ and step size $\eta^*$ gives
\[
\begin{aligned}
\RegSDAC{T}(\infty)
&=\sup_{\mathbf{m}\in\mathcal{M}_\infty}
\bigl(J_T^{\mathrm{on}}-J_T(\mathbf{m})\bigr)\\
&\le\RegSDAC{T}(H_T)
+\left(L_u+\frac{L_x}{1-\alpha}\right)\alpha^{H_T+1}T\\
&\le2\sqrt{\Gamma_\alpha T\log(H_T+1)}\\
&\quad+\left(L_u+\frac{L_x}{1-\alpha}\right)\alpha^{H_T+1}T,
\end{aligned}
\]
which proves \eqref{eq:infinite-sdac-regret-bound}.

For the asymptotic rate, the definition of $H_T$ gives
\[
\frac{\log T}{\log(1/\alpha)}
\le H_T
\le 1+\frac{\log T}{\log(1/\alpha)}.
\]
Since $0<\alpha<1$, it follows that
\[
\begin{aligned}
\alpha^{H_T+1}T
&\le\alpha^{1+\log T/\log(1/\alpha)}T
=\alpha\exp\!\left(
\frac{\log\alpha}{\log(1/\alpha)}\log T\right)T \\
&=\alpha e^{-\log T}T=\alpha.
\end{aligned}
\]
Also, for fixed $\alpha$,
\[
\log(H_T+1)
\le\log\!\left(2+\frac{\log T}{\log(1/\alpha)}\right)
=O(\log\log T).
\]
Thus, for fixed $\alpha,L_x,L_u$, the right-hand side of
\eqref{eq:infinite-sdac-regret-bound} is
$O\!\left(\sqrt{T\log\log T}\right)=o(T)$.
For fixed-fraction policies, Corollary~\ref{cor:fraction-as-sdac} gives
$\mathbf{m}(k)\in\mathcal{M}_\infty$ for every $k\in[0,\alpha]$. Consequently,
\[
\begin{aligned}
\inf_{\mathbf{m}\in\mathcal{M}_\infty}J_T(\mathbf{m})
&\le\min_{k\in[0,\alpha]}J_T(\mathbf{m}(k))\\
&=\min_{k\in[0,\alpha]}\sum_{t=1}^T c_t\bigl(x_t(k),u_t(k)\bigr).
\end{aligned}
\]
Subtracting from $J_T^{\mathrm{on}}$ yields the stated fixed-fraction regret bound.
\end{proof}

\section{Proofs for the Minimax Regret Lower Bound}
\label{sec:appendix-minimax-proofs}

\begin{proof}[Proof of Theorem~\ref{thm:sdac-lower-bound}]
Fix a causal feasible controller $\mathcal{A}$. Let
$\xi_1,\ldots,\xi_T$ be independent Rademacher random variables, independent
of the controller's randomization. Set
\[
w_t=1,
\qquad
c_t(x,u)=\xi_t(L_xx-L_uu),
\qquad t\in[T].
\]
Each cost is affine with gradient $\xi_t(L_x,-L_u)$, so it is convex and
satisfies the required subgradient bounds. Also,
\[
\begin{aligned}
|c_t(x,u)-c_t(x',u')|
&=|L_x(x-x')-L_u(u-u')|\\
&\le L_x|x-x'|+L_u|u-u'|.
\end{aligned}
\]
Thus, every realization gives an admissible sequence.

Consider the two fixed SDAC parameters
\[
\mathbf{m}^{(0)}=\mathbf{0},
\qquad
\mathbf{m}^{(1)}=\mathbf{e}_1=(1,0,\ldots,0)\in\mathcal{M}_H.
\]
Both trajectories start at $(x_1,u_1)=(0,0)$. For $t\ge2$, the SDAC
action formula and the dynamics give
\[
\begin{aligned}
u_t(\mathbf{m}^{(0)})&=0,
&
x_t(\mathbf{m}^{(0)})
&=\sum_{s=0}^{t-2}\alpha^s
=\frac{1-\alpha^{t-1}}{1-\alpha},\\
u_t(\mathbf{m}^{(1)})&=\alpha,
&
x_t(\mathbf{m}^{(1)})&=1.
\end{aligned}
\]
The last identity follows from $x_2(\mathbf{m}^{(1)})=1$ and
\[
x_{t+1}(\mathbf{m}^{(1)})
=\alpha x_t(\mathbf{m}^{(1)})-\alpha+1
=\alpha-\alpha+1=1
\]
whenever $x_t(\mathbf{m}^{(1)})=1$. Hence, for $t\ge2$,
\[
\begin{aligned}
x_t(\mathbf{m}^{(0)})-x_t(\mathbf{m}^{(1)})
&=\frac{1-\alpha^{t-1}}{1-\alpha}-1
=\frac{\alpha(1-\alpha^{t-2})}{1-\alpha}.
\end{aligned}
\]
Since both policies incur zero cost at $t=1$, their cumulative costs satisfy
\[
\begin{aligned}
J_T(\mathbf{m}^{(0)})-J_T(\mathbf{m}^{(1)})
&=\sum_{t=2}^T\xi_t
\left(
\frac{\alpha L_x(1-\alpha^{t-2})}{1-\alpha}
+\alpha L_u
\right)\\
&=\sum_{t=2}^T\xi_t\beta_t.
\end{aligned}
\]

By causality, $(x_t,u_t)$ depends only on past signs and the controller's
randomization, so it is independent of $\xi_t$. Since $\mathbb{E}[\xi_t]=0$,
\[
\begin{aligned}
\mathbb{E}_{\boldsymbol{\xi},\mathcal{A}}[J_T^{\mathrm{on}}]
&=\sum_{t=1}^T
\mathbb{E}[\xi_t]\,
\mathbb{E}_{\boldsymbol{\xi},\mathcal{A}}[L_xx_t-L_uu_t]
=0.
\end{aligned}
\]
The two fixed-policy trajectories are deterministic. Thus, for $j\in\{0,1\}$,
\[
\begin{aligned}
\mathbb{E}_{\boldsymbol{\xi}}[J_T(\mathbf{m}^{(j)})]
&=\sum_{t=1}^T\mathbb{E}[\xi_t]
\bigl(L_xx_t(\mathbf{m}^{(j)})-L_uu_t(\mathbf{m}^{(j)})\bigr)\\
&=0.
\end{aligned}
\]
Both parameters belong to $\mathcal{M}_H$. Using
$\min\{a,b\}=(a+b-|a-b|)/2$, we obtain
\[
\begin{aligned}
\RegSDAC{T}(H)
&=J_T^{\mathrm{on}}-\min_{\mathbf{m}\in\mathcal{M}_H}J_T(\mathbf{m})\\
&\ge J_T^{\mathrm{on}}
-\min\{J_T(\mathbf{m}^{(0)}),J_T(\mathbf{m}^{(1)})\}\\
&=J_T^{\mathrm{on}}
-\frac{J_T(\mathbf{m}^{(0)})+J_T(\mathbf{m}^{(1)})}{2}\\
&\quad+\frac12
\left|J_T(\mathbf{m}^{(0)})-J_T(\mathbf{m}^{(1)})\right|.
\end{aligned}
\]
Taking expectations and applying Khintchine's inequality in the last line gives
\[
\begin{aligned}
\mathbb{E}_{\boldsymbol{\xi},\mathcal{A}}[\RegSDAC{T}(H)]
&\ge\frac12\mathbb{E}_{\boldsymbol{\xi}}
\left|J_T(\mathbf{m}^{(0)})-J_T(\mathbf{m}^{(1)})\right| \\
&=\frac12\mathbb{E}_{\boldsymbol{\xi}}
\left|\sum_{t=2}^T\xi_t\beta_t\right|
\ge\frac{1}{2\sqrt2}
\sqrt{\sum_{t=2}^T\beta_t^2}.
\end{aligned}
\]
Averaging over $\boldsymbol{\xi}\in\{-1,1\}^T$ shows that some
$\boldsymbol{\xi}^\star$ gives expected regret at least this bound,
where the expectation is over the controller's randomization.
Fixing $\boldsymbol{\xi}^\star$ before the interaction gives the required
deterministic oblivious sequence.
\end{proof}

\begin{proof}[Proof of Corollary~\ref{cor:sdac-lower-mixing}]
For $t=t_{\mathrm{mix}}(\alpha)+2,\ldots,T$,
\[
\alpha^{t-2}
\le\alpha^{t_{\mathrm{mix}}(\alpha)}
\le
\alpha^{\log 2/\log(1/\alpha)}
=\frac12.
\]
Hence, by \eqref{eq:sdac-lower-beta},
\[
\begin{aligned}
\beta_t
&\ge
\alpha\left(
L_u+\frac{L_x}{2(1-\alpha)}
\right)
\ge
\frac{\alpha}{2}
\left(
L_u+\frac{L_x}{1-\alpha}
\right)
=\frac{G_\alpha}{2}.
\end{aligned}
\]
There are $T-t_{\mathrm{mix}}(\alpha)-1\ge T/2$ such indices, since
$T\ge2t_{\mathrm{mix}}(\alpha)+2$. Theorem~\ref{thm:sdac-lower-bound} therefore gives
\begin{align*}
\mathbb{E}_{\mathcal{A}}\!\left[\RegSDAC{T}(H)\right]
&\ge
\frac{1}{2\sqrt{2}}\sqrt{\sum_{t=2}^T\beta_t^2} \\
&\ge
\frac{1}{2\sqrt{2}}
\left[
\bigl(T-t_{\mathrm{mix}}(\alpha)-1\bigr)
\left(\frac{G_\alpha}{2}\right)^2
\right]^{1/2} \\
&\ge
\frac{1}{2\sqrt{2}}
\left[
\frac{T}{2}\left(\frac{G_\alpha}{2}\right)^2
\right]^{1/2}
=
\frac{G_\alpha}{8}\sqrt{T}.\qedhere
\end{align*}
\end{proof}

\IEEEtriggeratref{15}
\bibliographystyle{IEEEtran}
\bibliography{ref}

\begin{thebibliography}{10}
\providecommand{\url}[1]{#1}
\csname url@samestyle\endcsname
\providecommand{\newblock}{\relax}
\providecommand{\bibinfo}[2]{#2}
\providecommand{\BIBentrySTDinterwordspacing}{\spaceskip=0pt\relax}
\providecommand{\BIBentryALTinterwordstretchfactor}{4}
\providecommand{\BIBentryALTinterwordspacing}{\spaceskip=\fontdimen2\font plus
\BIBentryALTinterwordstretchfactor\fontdimen3\font minus
  \fontdimen4\font\relax}
\providecommand{\BIBforeignlanguage}[2]{{%
\expandafter\ifx\csname l@#1\endcsname\relax
\typeout{** WARNING: IEEEtran.bst: No hyphenation pattern has been}%
\typeout{** loaded for the language `#1'. Using the pattern for}%
\typeout{** the default language instead.}%
\else
\language=\csname l@#1\endcsname
\fi
#2}}
\providecommand{\BIBdecl}{\relax}
\BIBdecl

\bibitem{agarwal2019online}
N.~Agarwal, B.~Bullins, E.~Hazan, S.~Kakade, and K.~Singh, ``Online control
  with adversarial disturbances,'' in \emph{Proceedings of the 36th
  International Conference on Machine Learning}, ser. Proceedings of Machine
  Learning Research, vol.~97.\hskip 1em plus 0.5em minus 0.4em\relax PMLR,
  2019, pp. 111--119.

\bibitem{golowich2024population}
\BIBentryALTinterwordspacing
N.~Golowich, E.~Hazan, Z.~Lu, D.~Rohatgi, and Y.~J. Sun, ``Online control in
  population dynamics,'' in \emph{Advances in Neural Information Processing
  Systems}, vol.~37, 2024, pp. 111\,571--111\,613. [Online]. Available:
  \url{https://papers.nips.cc/paper_files/paper/2024/hash/ca4f6e86453e4b117dd3263792053cf5-Abstract-Conference.html}
\BIBentrySTDinterwordspacing

\bibitem{li2021onlineConstrained}
\BIBentryALTinterwordspacing
Y.~Li, S.~Das, and N.~Li, ``Online optimal control with affine constraints,''
  \emph{Proceedings of the AAAI Conference on Artificial Intelligence},
  vol.~35, no.~10, pp. 8527--8537, 2021. [Online]. Available:
  \url{https://ojs.aaai.org/index.php/AAAI/article/view/17035}
\BIBentrySTDinterwordspacing

\bibitem{liu2023constrainedControl}
\BIBentryALTinterwordspacing
X.~Liu, Z.~Yang, and L.~Ying, ``Online nonstochastic control with adversarial
  and static constraints,'' in \emph{Proceedings of the 40th International
  Conference on Machine Learning}, ser. Proceedings of Machine Learning
  Research, vol. 202.\hskip 1em plus 0.5em minus 0.4em\relax PMLR, 2023, pp.
  22\,277--22\,288. [Online]. Available:
  \url{https://proceedings.mlr.press/v202/liu23at.html}
\BIBentrySTDinterwordspacing

\bibitem{kumar2023unbounded}
\BIBentryALTinterwordspacing
R.~Kumar, S.~Dean, and R.~Kleinberg, ``Online convex optimization with
  unbounded memory,'' 2024, arXiv:2210.09903, version 5. [Online]. Available:
  \url{https://arxiv.org/abs/2210.09903}
\BIBentrySTDinterwordspacing

\bibitem{shaviv2016universally}
\BIBentryALTinterwordspacing
D.~Shaviv and A.~{\"O}zg{\"u}r, ``Universally near optimal online power control
  for energy harvesting nodes,'' \emph{IEEE Journal on Selected Areas in
  Communications}, vol.~34, no.~12, pp. 3620--3631, 2016. [Online]. Available:
  \url{https://arxiv.org/abs/1511.00353}
\BIBentrySTDinterwordspacing

\bibitem{arafa2018online}
A.~Arafa, A.~Baknina, and S.~Ulukus, ``Online fixed fraction policies in energy
  harvesting communication systems,'' \emph{IEEE Transactions on Wireless
  Communications}, vol.~17, no.~5, pp. 2975--2986, 2018.

\bibitem{zibaeenejad2019lookahead}
\BIBentryALTinterwordspacing
A.~Zibaeenejad and J.~Chen, ``The optimal power control policy for an energy
  harvesting system with look-ahead: Bernoulli energy arrivals,'' in \emph{2019
  IEEE International Symposium on Information Theory (ISIT)}, 2019, extended
  version: arXiv:1904.12281. [Online]. Available:
  \url{https://arxiv.org/abs/1904.12281}
\BIBentrySTDinterwordspacing

\bibitem{Yu2019}
H.~Yu and M.~J. Neely, ``Learning-aided optimization for energy-harvesting
  devices with outdated state information,'' \emph{IEEE/ACM Transactions on
  Networking}, vol.~27, no.~4, pp. 1501--1514, 2019.

\bibitem{10.1145/3428337}
\BIBentryALTinterwordspacing
K.~Asgari and M.~J. Neely, ``Bregman-style online convex optimization with
  energy harvesting constraints,'' \emph{Proceedings of the ACM on Measurement
  and Analysis of Computing Systems}, vol.~4, no.~3, Nov. 2020. [Online].
  Available: \url{https://doi.org/10.1145/3428337}
\BIBentrySTDinterwordspacing

\bibitem{lin2025optimal}
Q.~Lin, J.~Su, and M.~Chen, ``Optimal algorithms for online
  {Age-of-Information} optimization in energy harvesting systems,'' \emph{IEEE
  Transactions on Networking}, vol.~33, no.~6, pp. 3146--3161, 2025.

\bibitem{tassiulas1990stability}
L.~Tassiulas and A.~Ephremides, ``Stability properties of constrained queueing
  systems and scheduling policies for maximum throughput in multihop radio
  networks,'' in \emph{29th IEEE Conference on Decision and Control}.\hskip 1em
  plus 0.5em minus 0.4em\relax IEEE, 1990, pp. 2130--2132.

\bibitem{tassiulas1993dynamic}
------, ``Dynamic server allocation to parallel queues with randomly varying
  connectivity,'' \emph{IEEE Transactions on Information Theory}, vol.~39,
  no.~2, pp. 466--478, 1993.

\bibitem{mckeown1999achieving}
N.~McKeown, A.~Mekkittikul, V.~Anantharam, and J.~Walrand, ``Achieving 100\%
  throughput in an input-queued switch,'' \emph{IEEE Transactions on
  Communications}, vol.~47, no.~8, pp. 1260--1267, 1999.

\bibitem{neely2010stochastic}
M.~J. Neely, \emph{Stochastic Network Optimization with Application to
  Communication and Queueing Systems}.\hskip 1em plus 0.5em minus 0.4em\relax
  Morgan \& Claypool Publishers, 2010.

\bibitem{khojastepour2004delay}
M.~A. Khojastepour and A.~Sabharwal, ``Delay-constrained scheduling: Power
  efficiency, filter design, and bounds,'' in \emph{IEEE INFOCOM 2004},
  vol.~3.\hskip 1em plus 0.5em minus 0.4em\relax IEEE, 2004, pp. 1938--1949.

\bibitem{yang2020online}
L.~Yang, M.~H. Hajiesmaili, R.~Sitaraman, A.~Wierman, E.~Mallada, and W.~S.
  Wong, ``Online linear optimization with inventory management constraints,''
  \emph{Proceedings of the ACM on Measurement and Analysis of Computing
  Systems}, vol.~4, no.~1, 2020.

\bibitem{hihat2024online}
\BIBentryALTinterwordspacing
M.~Hihat and A.~Fermanian, ``Online policy selection for inventory problems,''
  \emph{arXiv preprint arXiv:2411.19269}, 2024. [Online]. Available:
  \url{https://arxiv.org/abs/2411.19269}
\BIBentrySTDinterwordspacing

\bibitem{moran1954dams}
P.~A.~P. Moran, ``A probability theory of dams and storage systems,''
  \emph{Australian Journal of Applied Science}, vol.~5, no.~2, pp. 116--124,
  1954.

\bibitem{rockafellar1998variational}
R.~T. Rockafellar and R.~J.-B. Wets, \emph{Variational Analysis}, ser.
  Grund\-leh\-ren der ma\-the\-ma\-ti\-schen Wis\-sen\-schaf\-ten.\hskip 1em
  plus 0.5em minus 0.4em\relax Berlin: Springer, 1998, vol. 317.

\bibitem{hazan2016introduction}
E.~Hazan, ``Introduction to online convex optimization,'' \emph{Foundations and
  Trends in Optimization}, vol.~2, no. 3--4, pp. 157--325, 2016.

\end{thebibliography}

\end{document}